\documentclass{sig-alternate}
\usepackage{hyperref,breakurl}
\usepackage{blindtext}
\usepackage{algorithm} 
\usepackage{algorithmic} 
\usepackage{color}
\usepackage{url}
\usepackage{graphicx}
\usepackage{subfigure}
\allowdisplaybreaks

\newtheorem{theorem}{Theorem}
\newtheorem{lemma}{Lemma}
\newtheorem{proposition}{Proposition}

\newcommand{\ud}{\mathrm{d}}
\newcommand{\calP}{\mathcal{P}}
\newcommand{\calQ}{\mathcal{Q}}

\begin{document}
%
\conferenceinfo{SIGMETRICS}{2014 Austin, Texas USA}

\title{The Multi-shop Ski Rental Problem}
%
%
%
%
%

\numberofauthors{6} 
%

\newcounter{savecntr}
\newcounter{restorecntr}

\author{
%
%
\alignauthor
Lingqing Ai \\
       \affaddr{IIIS, Tsinghua University}\\
       \email{ailingqing@126.com}
\alignauthor
Xian Wu \\
       \affaddr{IIIS, Tsinghua University}\\
       \email{wuxian12@\\
       mails.tsinghua.edu.cn}
\alignauthor Lingxiao Huang \\
       \affaddr{IIIS, Tsinghua University}\\
       \email{huanglingxiao1990@126.com}
\and  
\alignauthor Longbo Huang\\
       \affaddr{IIIS, Tsinghua University}\\
       \email{longbohuang@ \\
       tsinghua.edu.cn }
\alignauthor Pingzhong Tang\\
       \affaddr{IIIS, Tsinghua University}\\
       \email{kenshin@tsinghua.edu.cn}
\alignauthor Jian Li\\
       \affaddr{IIIS, Tsinghua University}\\
       \email{lijian83@tsinghua.edu.cn}
}

\maketitle
\begin{abstract}
We consider the {\em multi-shop ski rental} problem. This problem generalizes the classic ski rental problem to a multi-shop setting, in which each shop has different prices for renting and purchasing a pair of skis, and a \emph{consumer} has to make decisions on when and where to buy. We are interested in the {\em optimal online (competitive-ratio minimizing) mixed strategy} from the consumer's perspective. For our problem in its basic form, we obtain exciting closed-form solutions and a linear time algorithm for computing them. We further demonstrate the generality of our approach by investigating three extensions of our basic problem, namely ones that consider costs incurred by entering a shop or switching to another shop.  Our solutions to these problems suggest that the consumer must assign positive probability in \emph{exactly one} shop at any buying time. Our results apply to many real-world applications, ranging from cost management in \texttt{IaaS} cloud to scheduling in distributed computing.
\end{abstract}



\section{INTRODUCTION}

The ski rental problem ($\mathtt{SR}$) is a dilemma faced by a consumer, who is uncertain about how many days she will ski and has to  trade off between buying and renting skis: once she buys the skis, she will enjoy the remaining days rent-free, but before that she must pay the daily renting cost. The literature is interested in investigating the {\em online optimal strategy} of the consumer. That is, a strategy that yields the lowest competitive ratio \emph{without having any information of the future} (as is standard in the literature, competitive ratio is defined as the the ratio between the cost yielded by the consumer's strategy and the cost yielded by the optimal strategy of a prophet, who foresees how many days the trip  will last and design the optimal strategy accordingly). The ski rental problem and its variants constitute an important part of the online algorithm design literature from both theoretical and applied perspectives~\cite{fleischer2001bahncard,karlin2001dynamic,lin2011dynamic,lotker2008rent,lu2012simple,wang2013reserve}.

In this paper, we consider the  \emph{multi-shop ski rental problem}($\texttt{MSR}$), in which the consumer faces multiple shops that offer different renting and buying prices. She must choose one shop immediately after she arrives at the ski field and must rent or buy the skis \emph{in that particular shop} since then.  In other words, once she has chosen a shop, the only decision variable is when to buy the skis. Beyond the basic setting,  we also propose three important extensions of \texttt{MSR} as below:
\begin{itemize}
\item $\texttt{MSR}$ \emph{with switching cost} ($\texttt{MSR-S}$): The consumer is allowed to switch from one shop to another and each switching costs her some constant amount of money.
\item $\texttt{MSR}$ \emph{with entry fee} ($\texttt{MSR-E}$): Each shop requires some entry fee and the consumer \emph{cannot}  switch shops.
\item $\texttt{MSR}$ \emph{with entry fee and switching} ($\texttt{MSR-ES}$): The consumer is able to switch from one shop to another, and she pays the entry fee as long as she enters any shop\footnote{For example, if she switches from shop 1 to shop 2, and then switches back to shop 1, she pays the entry fee of shop 1 twice and the entry fee of shop 2 once.}.
\end{itemize}

In all the settings above, the consumer's objective is to minimize the competitive ratio. In \texttt{MSR} and \texttt{MSR-E}, she has to consider two questions \emph{at the very beginning}: (1) where should she rent or buy the skis (place), and (2) when should she buy the skis (timing)? While \texttt{MSR-S} and \texttt{MSR-ES} allow the consumer to switch shops and are thus more fine-grained than the previous two, in the sense that she is able to decide where to rent or buy the skis \emph{at any time}. For example, it is among her options to rent in shop 1 on day 1, switch to shop 2 from day 2, and finally switch to shop 3  and then buys the skis.

The multi-shop ski rental problem naturally extends the ski rental problem and allows  heterogeneity in consumer's options, a desirable feature that makes the ski rental problem a more general modeling framework for online algorithm design.
Below, we present a few  real world scenarios that can be modeled with the  multi-shop ski rental problem.

\textbf{1. Scheduling in distributed computing:} A file is replicated and stored in different machines in the cluster. Some node undertaking some computing job needs data in the file during the execution. The node can either request the corresponding data block of the file from some selected machine whenever it needs to, which incurs some delay, or it can simply ask that machine to transmit the whole file beforehand, at the sacrifice of a longer delay at the beginning without any further waiting.  When selecting the replicating machine, the scheduling node needs to consider the current bandwidth, read latency, etc.  In this application, each replicating machine is considered as a shop, and renting corresponds to requesting for the data block on-demand while buying means to fetch the whole file beforehand.

\textbf{2. Cost management in IaaS cloud:} Multiple \texttt{IaaS} cloud vendors, such as Amazon EC2~\cite{EC2}, ElasticHosts~\cite{ElasticHosts} and Microsoft Windows Azure~\cite{Azure}, offer different price options, which can be classified into two commitment levels: users pay for \emph{on-demand} server instances at an hourly rate or make a one-time, upfront payment to host each instances for some duration (e.g., monthly or annually), during which users either use the instances for free~\cite{ElasticHosts}, or enjoy a discount on renting the instance~\cite{EC2}.

Consider an example in Table 1. Table 1 lists the pricing options for the instances with identical configurations offered by Amazon EC2 and ElasticHosts. Each pricing option can be considered as a shop in the multi-shop ski rental problem, where in the 1(3) year(s) term contract in Amazon EC2, the entry fee is the upfront payment and the hourly price is the renting price.

\begin{table}[hbt]
\begin{tabular}{|c|c|c|c|}
  \hline
  Vendor& Option & Upfront(\$) & Hourly(\$)\\ \hline
  & On-Demand & 0& 0.145\\ \cline{2-4}
  Amazon & 1 yr Term & 161 & 0.09 \\ \cline{2-4}
   & 3 yr Term & 243 & 0.079 \\ \hline
 ElasticHosts & 1 mo Term &  97.60  & 0 \\ \cline{2-4}
    & 1 yr Term & 976.04 & 0 \\ \hline
\end{tabular}
\caption{Pricing Options of the `same' instance in Amazon EC2 (c1.medium) and ElasticHosts (2400MHz cpu, 1792MHz memory, 350Gb storage and no data transfer).}
\end{table}

\textbf{3. Purchase decisions}: A company offering high-resolution aerial or satellite map service chooses between Astrium~\cite{Astrium} and DigitalGlobe~\cite{DigitalGlobe}. It can either subscribe imagery from one company or exclusively occupy the service by `purchasing' one satellite like what Google has done~\cite{News}. Similar applications include some person purchasing a SIM card from different telecommunication companies.

\subsection{Related Work}

The ski rental problem is first considered by Karlin \emph{et al.} \cite{karlin1988competitive}, and then studied by Karlin's seminal paper~\cite{karlin1994competitive} which proposes a randomized algorithm and gives a $\frac{e}{e-1}$ competitive ratio. Later researchers propose a few variants, including the Bahncard problem \cite{fleischer2001bahncard} and the TCP acknowledgment problem~\cite{karlin2001dynamic}. A more recent work~\cite{khanafer2013constrained} analyzes the case in which the first or the second moment of the skiing days are known and gives an optimal online solution. However, all the aforementioned works deal with the case where a single consumer rents or buys the skis in one single shop. In their problems,  the consumer only needs to decide when to buy.  While in the multi-shop ski rental problem, the consumer has to make a two-fold decision (time and place).  Closest to our work is the work by Lotker \emph{et al.}~\cite{lotker2008rent}, which considers the case where  the consumer has multiple options in one shop, i.e., the multi-slop problem  and their problem can be regarded as a special case of our problem by setting all the buying prices sufficiently large.

Research on ``multiple consumers in one single shop'' have been conducted from applied perspectives~\cite{lin2011dynamic,lu2012simple,wang2013reserve}.  Lin \emph{et al.}~\cite{lin2011dynamic} investigate a  dynamical `right-sizing' strategy by turning off servers during periods of low load. Their model and lazy capacity provisioning algorithms closely tie to the ski rental problem.  Lu \emph{et al.}~\cite{lu2012simple} derive the ``dynamic provisioning techniques'' to turn on or off the servers to minimize the energy consumption. They dispatch the servers so that each individual server is reduced to a standard ski rental problem. Wei \emph{et al.}~\cite{wang2013reserve} propose an online algorithm to serve the time-varying demands at the minimum cost in \texttt{IaaS} cloud given one price option, in which the `consumers' (servers) may be related to each other.

Another line of work \cite{bodenstein2011strategic,hong2011dynamic} focusing on minimizing the cost in data centers or other cloud services  assumes that the long-term workloads are stationary and thus can be predicted, and Guenter \emph{et al.}~\cite{guenter2011managing}   consider the cases of short predictions.
However,  for many real-world applications,  future workloads can exhibit non-stationarity \cite{singh2010autonomic}. Other researchers~\cite{lin2011dynamic,lu2012simple,wang2013data,wang2013reserve} that require \emph{no priori} knowledge of the future workload minimize the cost \emph{given one option is selected}. Our paper is orthogonal to theirs since we  focus on how to select a price option.

\subsection{Our Contributions}

In this paper, we consider the multi-shop ski-rental problem and its extensions, in which there are multiple shops and the consumer must make two-fold decisions (time and place) to minimize the competitive ratio. We model each problem using a zero-sum game played by the consumer and nature. We simplify the strategy space of the consumer via \emph{ removal of strictly dominated strategies} and derive the form of the optimal mixed strategy of the consumer. We summarize the key contributions as follows:
\begin{enumerate}
\item[1.] For each of the problems, we prove that under the optimal mixed strategy of the consumer, the consumer only assigns positive buying probability to \emph{exactly one} shop at any time. As the buying time increases, she follows the shop order in which the ratio between buying price and renting price is increasing. This order also holds in \texttt{MSR-E} and \texttt{MSR-ES}, where entry fee is involved.

\item[2.] We derive a novel, easy-to-implement \emph{linear time} algorithm for computing the optimal strategy of the consumer, which drastically reduces the complexity of computing the solution to \texttt{MSR}.

\item[3.] For \texttt{MSR-S}, we prove that under the optimal mixed strategy, the consumer only needs  to consider switching to another shop at the buying time, i.e., she will never switch to another shop and continue renting. Moreover, we show that \texttt{MSR-S} can be reduced to an equivalent \texttt{MSR} problem with modified buying prices.

\item[4.] For \texttt{MSR-ES}, we prove that under the optimal mixed strategy, the consumer may switch to another shop either during the renting period or at the buying time, but she only follows some particular order of switching. Moreover, the  number of times of switching is no more $n$ where $n$ is the number of shops.

\item[5.] We characterize any action of the consumer in \texttt{MSR-ES} by proving that the action can be decoupled into a sequence of  operations. We further show that each operation can be viewed as a virtual shop in \texttt{MSR-E} and in total, we create $O(n^2)$ `virtual' shops of \texttt{MSR-E}. Therefore, \texttt{MSR-ES} can be reduced to \texttt{MSR-E} with minor modifications.

\end{enumerate}

\section{BASIC PROBLEM}

In the \emph{multi-shop ski rental problem} (also \texttt{MSR}), a person goes skiing for an anbiguous time period. There are multiple shops providing skies either for rental or for buying. The person must choose one shop \emph{as soon as} she arrives at the ski field, and she can decide whether or not to buy the skis in that particular shop at any time\footnote{ In this paper, we focus on the continuous time model.}. Note that she cannot change the shop once she chooses one. The objective is to minimize the worst-case ratio between the amount she actually pays and the money she would have paid if she knew the duration of skiing in advance. We assume that there are $n$ shops in total, denoted by $[n]\triangleq\{1,2,3,\cdots,n\}$. Each shop $j$ offers skis at a renting price of $r_j$ dollars per unit time and at a buying price of $b_j$ dollars. This problem is a natural extension of the classic ski rental problem (\texttt{SR}) and  it is exactly \texttt{SR}  when $n=1$.

In \texttt{MSR}, it is clear that if there is a shop of which the rental and buying prices are both larger than those of another shop, it is always suboptimal to choose this shop. We assume that
\begin{align}
0<&~r_1<r_2<\cdots<r_n \nonumber\\
&~b_1>b_2>\cdots>b_n>0 \nonumber
\end{align}

We apply a game-theoretic approach for solving our problem. For the case of expressing the formulation, we assume that how long the consumer skis is determined by a player called \emph{nature}. Therefore, there are two parties in the problem, and we focus on the optimal strategy of the consumer.

In the remainder of this section, we first formulate our problem as a zero-sum game, and simplify the strategy space in Lemma~\ref{lemma:MSRstspace}. Then, we combine Lemma \ref{lemma:MSRconstant}-\ref{lemma:AlphaRelation}, and fully characterize the optimal strategy of the consumer in Theorem~\ref{theorem:MSR}. We show that optimally the consumer assigns positive buying probability to \emph{exactly one} shop at any time. Moreover, the possible times the consumer buys the skis in a shop constitute a continuous interval.
Thus, we can partition the optimal strategy of the consumer into different sub-intervals which relate to different shops, and the problem is reduced to how to find the optimal breakpoints. Based on Lemma~\ref{lemma:dnconcave} and~\ref{lemma:GConcave}, we develop a linear time algorithm for computing the optimal breakpoints and prove its correctness in Theorem~\ref{theorem:MSRalg}.

\subsection{Formulation}

We first analyze the action set for both players in the game. For the consumer, we denote by $j$ the index of the shop in which she rents or buys the skis. Let $x$ be the time when she chooses to buy the skis, i.e., the consumer will rent the skis before $x$ and buy at $x$ if nature has not yet stopped her. The action of the consumer is thus represented by a pair $(j,x)$. Denote by $\Psi_c$ the action set of the consumer:
\begin{displaymath}
\Psi_c \triangleq \{(j,x): j \in [n], x\in [0,+\infty)\cup\{+\infty\} \}
\end{displaymath}
where $x =+\infty$ means that the consumer always rents and never buys. Next, let $y$ denote the time when nature stops the consumer from skiing. Thus, the action set of nature is
\begin{displaymath}
\Psi_n \triangleq \{y: y \in (0,+\infty)\cup\{+\infty\} \}
\end{displaymath}
where $y =+\infty$ means that nature always allows the consumer to keep skiing.
If $y=x$, we regard it as the case that right after the consumer buys the skis, nature stops her.
Given the strategy profile $\langle(j,x),y\rangle$, let $c_j(x,y) \geq 0$  denote the cost paid by the consumer:
\begin{displaymath}
c_j(x,y) \triangleq
\begin{cases}
r_j y, & y < x\\
r_j x + b_j, & y \geq x
\end{cases}
\end{displaymath}

Now we define the strategy space for the consumer and nature. Let $\mathbf{p} \triangleq (p_1,\cdots,p_n)$ be a mixed strategy represented by a vector of probability density functions. $p_j(x)$ is the density assigned to the strategy $(j,x)$ for any $j = 1,\cdots, n$ and $x\in [0,+\infty)\cup\{+\infty\}$. In this paper, we assume that for each point, either probability density function exists or it is probability mass.\footnote{In fact, our results can be extended to the case where in the strategy space the cumulative distribution function is not absolutely continuous and thus no probability density function exists.} If $p_j(x)$ is probability mass, we regard $p_j(x)$ as $+\infty$ and define $p_{j,x}\triangleq\int_{x^-}^{x} p_j(t)dt$ satisfying $p_{j,x}\in(0,1]$. The strategy space $\calP$ of the consumer is as follows:
\footnote{For convenience, we denote by $\int_{a}^{b}f(x)dx$ ($a<b$) the integral over $(a,b]$, except that when $a=0$, the integral is over $[0,b]$.}
\begin{eqnarray*}
\calP = \Bigg\{\mathbf{p}: && \sum_{j=1}^n \int_{0}^\infty p_j(x) \ud x = 1,\\
 && p_j(x) \ge 0, \forall x\in [0,+\infty)\cup\{+\infty\}, \forall j\in [n] \Bigg\}
\end{eqnarray*}

Similarly, define $q(y)$ to be the probability density of nature choosing $y$ and the strategy space $\calQ$ of nature is given by

\begin{displaymath}
\calQ =\Bigg\{\mathbf{q}: \int_0^\infty q(y) \ud y = 1, q(y) \geq 0, \forall y \in (0,+\infty)\cup\{+\infty\} \Bigg\}
\end{displaymath}

When the consumer chooses the mixed strategy $\mathbf{p}$ and  nature chooses the stopping time $y$, the expected cost to the consumer is:
\begin{displaymath}
C(\mathbf{p},y)\triangleq\sum_{j=1}^n C_j(p_j,y)
\end{displaymath}
in which
\begin{eqnarray*}
C_j(p_j,y)&\triangleq&\int_{0}^{\infty}c_j(x,y)p_j(x)\ud x\\
&=&\int_{0}^{y}(r_jx+b_j)p_j(x)\ud x+\int_{y}^{\infty}yr_j p_j(x)\ud x
\end{eqnarray*}
is the expected payment to shop $j$ for all $j\in [n]$. Given the strategy profile $\langle\mathbf{p}, \mathbf{q}\rangle$, the competitive ratio is defined as:
\begin{eqnarray} \label{def:Ratio}
R(\mathbf{p}, \mathbf{q}) &\triangleq& \int_0^\infty \frac{C(\mathbf{p}, y)}{\mathrm{OPT}(y)} q(y) \ud y
\end{eqnarray}

Here $\mathrm{OPT}(y)$ is the optimal offline cost and can be seen to have the following form:
\begin{equation} \label{offlineOpt: MSR}
\mathrm{OPT}(y) =
\begin{cases}
r_1 y, & y \in (0,B]\\
b_n, & y > B
\end{cases}
\end{equation}
where $B$ is defined as $B\triangleq\frac{b_n}{r_1}$.

Note that $B$ is the dividing line between the minimum buying cost and the minimum renting cost. When $y<B$, the offline optimal is always to rent at the first shop, and when $y>B$, the offline optimal is to buy the skis at the last shop.  We will show that $B$ determines the effective action sets of the consumer and nature in section~\ref{sec:simplify} .

The objective of the consumer is to minimize the worst-case competitive ratio, i.e., to choose a strategy $\mathbf{p} \in \calP$ that solves the problem
\begin{eqnarray*}
&\mathrm{minimize}& \max_{y>0}  \left\{\frac{C(\mathbf{p},y)}{\mathrm{OPT}(y)} \right\} \\
&\textrm{subject to}& \mathbf{p} \in \calP
\end{eqnarray*}
which is equivalent to the following:
\begin{eqnarray}
&\mathrm{minimize}& \lambda \label{problem:MSR1}\\
&\textrm{subject to}& \frac{C(\mathbf{p}, y)}{r_1 y} \leq \lambda \nonumber\\
&& \sum_{j=1}^n \int_{0}^\infty p_j(x) \ud x = 1 \nonumber\\
&& p_j(x) \geq 0 \ \ \forall x\in [0,+\infty)\cup\{+\infty\}\nonumber\\
&& \forall y\in (0,+\infty)\cup\{+\infty\},  \forall j \in [n] \nonumber
\end{eqnarray}

\subsubsection{Simplifying the Zero-sum Game}\label{sec:simplify}

In this section, we show that the game can be greatly simplified and the action set for both the consumer and nature can be reduced. Specifically, nature prefers the strategy $y = +\infty$ to any other strategy $y' > B$. For the consumer, for any $j \in [n]$, she prefers the strategy $(j,B)$ to any other strategy $(j,x')$ where $x' > B$.

\begin{lemma}
\label{lemma:MSRstspace}
For nature, any strategy $y \in [B,+\infty)$ is dominated. While for the consumer, any strategy $(j,x)$ is dominated, in which $x \in (B,+\infty)\cup\{+\infty\}, \forall j \in [n]$.
\end{lemma}

\begin{proof}
Recall the cost $c_j(x,y)$ is defined as follows:
\begin{displaymath}
c_j(x,y) =
\begin{cases}
r_j y, & y < x\\
r_j x + b_j, & y \geq x
\end{cases}
\end{displaymath}
Thus for any fixed $(j,x)$, $c_j(x,y)$ is a non-decreasing function of $y$. Further, from (\ref{offlineOpt: MSR}), we can see that the offline optimal cost is unchanged when $y\geq B$. Thus, for any $y \geq B$ it holds that
\begin{displaymath}
\frac{c_j(x,y)}{b_n} \leq \lim_{y\rightarrow +\infty}\frac{c_j(x,y)}{b_n} = \frac{r_j x + b_j}{b_n}
\end{displaymath}
Therefore, any strategy of nature that includes $y \ge B$ is dominated by the strategy of never stopping the consumer.

Now for the consumer, for any shop $j \in \{1,\cdots, n\}$, and  any $x' \in (B,+\infty)\cup\{+\infty\}$, it holds that
\begin{displaymath}
c_j(B,y) - c_j(x',y) \leq 0, \quad \forall y \in (0,B)\cup \{ +\infty\}
\end{displaymath}
Therefore, any strategy of the consumer that includes buying at time $x'$ in any shop is dominated by the strategy of buying at $B$ in the same shop.
\end{proof}

From this lemma, the consumer's buying time is restricted in $[0,B]$. Note that for any $(j,x)$ in which $x \in [0,B]$, it holds that
\begin{displaymath}
\frac{c_j(x,B)}{\mathrm{OPT}(B)} = \frac{c_j(x,+\infty)}{\mathrm{OPT}(+\infty)}
\end{displaymath}
Therefore, the action set of nature $\Psi_n$ can be reduced to $\Psi_n = \{y \in (0,B]\}$.

Similarly, in the strategy space of the consumer $\calP$, nature $\calQ$, the expected cost $C(\mathtt{p},y)$ and the competitive ratio $R(\mathbf{p,q})$, we can replace $+\infty$ by $B$.

\emph{Comments on $B$}: recall that the boundary $B$ is defined as $\frac{\min\{b_i\}}{\min\{r_i\}} = \frac{b_n}{r_1}$ in \texttt{MSR}, while this value is $\frac{b_j}{r_j}$ if shop $j$ is the only shop in \texttt{SR}. For instance, if only shop $n$ appears in \texttt{SR}, then the consumer will never consider to buy at any time $x > \frac{b_n}{r_n}$. However, in \texttt{MSR}, the consumer may want to put some positive possibility to the strategy of buying at time $x>\frac{b_n}{r_n}$ in shop $n$ (since $r_1 < r_n$). The difference between these two cases is due to the fact that in \texttt{MSR}, the consumer has the global information of all the shops and the offline optimal is always to rent at shop 1 at the cost of $r_1$ per unit time until the total cost reaches the minimum buying price $b_n$, whereas in \texttt{SR}, the consumer always rents at the cost of $r_n\geq r_1$ per unit time until $b_n$.

With the above results, problem (\ref{problem:MSR1}) can now be reduced to the following:
\begin{align}
\mathrm{minimize}~& ~~~~~~~~~~\lambda \label{problem:MSR2}\\
\textrm{subject to}& ~~~~\frac{C(\mathbf{p}, y)}{r_1 y} \leq \lambda \tag{\ref{problem:MSR2}a}\\
& ~~~~\sum_{j=1}^n \int_{0}^B p_j(x) \ud x = 1 \tag{\ref{problem:MSR2}b}\\
& ~~~~p_j(x) \geq 0 \tag{\ref{problem:MSR2}c}\\
& ~~~~\forall x \in [0,B], \forall y \in (0,B], \forall j \in [n] \tag{\ref{problem:MSR2}d}
\end{align}

We will show that the optimal strategy of the consumer results in exact equality in (\ref{problem:MSR2}a) in the next subsection.

\subsection{Optimal Strategy of the Consumer}

In this subsection, we look into the optimal solution $\mathbf p^*$ for (\ref{problem:MSR2}). In short, $\mathbf p^*$ yields the same expected utility for nature whenever nature chooses to stop. In other words, given $\mathbf p^*$, any pure strategy of nature yields the same utility for both the consumer and nature. Moreover, at any time $x$, the consumer assigns positive buying probability to exactly one of the shops, say shop $j$, and $j$ is decreasing as $x$ increases. Finally, we can see that for any shop $j$ and the time interval in which the consumer chooses to buy at shop $j$, the density function $p_j(x)$ is $\alpha_j e^{r_j/b_j x}$ where $\alpha_j$ is some constant to be specified later.

We now state our first theorem that summarizes the structure of the optimal strategy.

\begin{theorem}
\label{theorem:MSR}
The optimal solution $\mathbf p^*$ satisfies the following properties:
\begin{itemize}
\item[(a)] There exists a constant $\lambda$, such that $\forall y \in (0,B]$,
\begin{displaymath}
\frac{C(\mathbf{p}^*,y)}{r_1 y} = \lambda
\end{displaymath}
\item[(b)] There exist $n+1$ breakpoints: $d_1,d_2,\cdots,d_{n+1}$, such that $B=d_1\ge d_2\ge\cdots\ge d_n\ge d_{n+1}=0$, and $\forall j\in [n]$, we have
\begin{equation*}
p_j^*(x)=
\begin{cases}
\alpha_j e^{r_jx/b_j},&x \in  (d_{j+1},d_j)\\
0,& otherwise
\end{cases}
\end{equation*}
in which $\alpha_j$ satisfies that
\begin{displaymath}
\alpha_j b_j e^{r_jd_{j}/b_j} = \alpha_{j-1} b_{j-1} e^{r_{j-1}d_{j}/b_{j-1}} \quad \forall j = 2,\cdots, n
\end{displaymath}
\end{itemize}
\end{theorem}

In the following, We will prove property (a) by Lemma~\ref{lemma:MSRconstant}, property (b) by Lemma~\ref{lemma:MSRfinite}-\ref{lemma:AlphaRelation}. All proof details can be found in Appendix A.

\begin{lemma}
\label{lemma:MSRconstant}
$\forall y \in (0,B]$, $\mathbf{p}^*$ satisfies that
\begin{equation}
\frac{C(\mathbf{p}^*,y)}{r_1 y} = \lambda \label{ratioRelation:MSR2}
\end{equation}
\end{lemma}

From the above lemma, the problem (\ref{problem:MSR2}) is thus equivalent to the following:
\begin{eqnarray}
&\mathrm{minimize}& \lambda \label{problem:MSR3}\\
&\textrm{subject to}& (\ref{ratioRelation:MSR2}),(\ref{problem:MSR2}b),(\ref{problem:MSR2}c),(\ref{problem:MSR2}d) \nonumber
\end{eqnarray}

Here are some intuitions of \texttt{MSR}: In the extreme case where the buying time $x$ is sufficiently small, the consumer will prefer shop $n$ than any other shops since $b_n$ is the minimum buying price. As $x$ increases, the renting cost weights more and the skier gradually chooses the shop with lower rent yet higher buying price. In the other extreme case when the skier decides to buy at time $x$ close to $B$, shop 1 may be the best place since it has the lowest rent. Thus,in the optimal strategy, the  interval $[0,B]$ may be partitioned into several sub-intervals. In each interval, the consumer only chooses to buy at one and only one shop. The following two lemmas formally show that the above intuitions are indeed the case.

\begin{lemma}
\label{lemma:MSRfinite}
$\forall j \in [n]$, we have $p^*_j(0)<+\infty$, and $\forall x \in (0,B]$, $p^*_j(x)< \frac{2b_1r_1}{b_n^2}$.
\end{lemma}

\begin{lemma}
\label{lemma:moveP MSR}
In the optimal strategy $\mathbf{p^*}$, there exists $n+1$ breakpoints $B=d_{1}\ge d_{2}\ge\cdots\ge d_{n+1} = 0$, which partition $[0,B]$ into $n$ sub-intervals, such that $\forall j=1,\cdots,n$, $\forall x \in (d_{j+1},d_{j})$, $p_j^*(x)>0 $ and $p_i^*(x)=0$ for any $i\neq j$.
\end{lemma}
\begin{proof} (sketch)
It suffices to show that $\forall x\in (0,B)$, $\forall \epsilon>0$, if there exists some $j$ such that $\int_{x-\epsilon}^{x}p_j^*(t)\ud t> 0$, then $\forall j'>j, x'\ge x$, we must have $\int_{x'}^{B}  p_{j'}^*(t)\ud t=0$. We use reductio ad absurdum to prove this proposition.

We first show that if there exists some $j'>j, x'>x, \epsilon>0$ such that $\int_{x-\epsilon}^{x}p_j^*(t)\ud t>0$, $\int_{x'}^{x'+\epsilon}p_{j'}^*(t)\ud t>0$, then there exist 2 intervals $(x_1,x_1+\theta)\subseteq(x-\epsilon,x)$ and $(x_2,x_2+\theta)\subseteq(x',x'+\epsilon)$, such that $$\int_{0}^{\epsilon_0} \min\{p_j^*(x_1+\theta),p_{j'}^*(x_2+\theta)\} \ud \theta>0$$

We next move some suitable buying probabilities of $p_{j'}^*$ from $(x_2,x_2+\theta)$ to $(x_1,x_1+\theta)$ for shop $j'$, and correspondingly move some purchase probabilities of $p_j^*$ from $(x_1,x_1+\theta)$ to $(x_2,x_2+\theta)$ for shop $j$. Then we obtain a new strategy $\mathbf{p^1}$. We show that $\forall y\in(0,B]$, $\mathbf{p^1}$ is no worse than $\mathbf{p^*}$, and $\forall y\in(x_1,B]$, $\mathbf{p^1}$ is strictly better than $\mathbf{p^*}$, which makes a contradiction.
\end{proof}

\vspace{-2mm}
The lemma explicitly specifies the order of the shops in the optimal strategy:  as $x$ increases, the index of the shop where the consumer assigns positive density decreases. Based on this lemma, for any $j \in [n], x \in (d_{j+1},d_j)$, multiplying both sides of (\ref{ratioRelation:MSR2}) by $r_1 y$, and taking twice derivatives, we have
\begin{equation} \label{Pdiffequation1}
 b_i\frac{\ud p_j^*(x)}{dx}=r_j p_j^*(x)\quad \forall x \in  (d_{j+1},d_{j})
\end{equation}
Solving this differentiable equation, we obtain the optimal solutions as follows\footnote{Because $p_j(x)$ is finite, we say $p_j(d_i) = 0$ for all $i,j\in N$, which does not affect the expected cost at all.}:
\begin{equation}\label{PFormSolution}
p_j^*(x)=
\begin{cases}
\alpha_j e^{r_jx/b_j},& x \in  (d_{j+1},d_j)\\
0,& otherwise
\end{cases}
\end{equation}
where $\alpha_j$ is some constant. The relationship between $\alpha_j$ and $\alpha_{j-1}$ is described in the following lemma:
\begin{lemma}
\label{lemma:AlphaRelation}
\begin{equation} \label{alphaRelation1}
\alpha_j b_j e^{r_jd_{j}/b_j} = \alpha_{j-1} b_{j-1} e^{r_{j-1}d_{j}/b_{j-1}} \quad \forall j = 2,\cdots, n
\end{equation}
\end{lemma}

\subsection{Computing the Optimal Strategy}\label{sec:computing}
In this section we propose a linear time algorithm to compute the optimal strategy for the consumer. First we show the relationship between the competitive ratio $\lambda$ and $\alpha_1$ by the following lemma:
\begin{lemma}
\label{lemma:MSRconratio}
For any strategy $\mathbf{p}$ which satisfies property (b) in Theorem~\ref{theorem:MSR}, it holds that
\begin{displaymath}
\frac{C(\mathbf p,y)}{r_1 y} = \alpha_1 \frac{b_1}{r_1} e^{\frac{r_1}{b_1}B},\quad\forall y\in(0,B]
\end{displaymath}
\end{lemma}
From the above lemma, we know that minimizing $\lambda$ is equivalent to minimizing $\alpha_1$. Therefore, problem (\ref{problem:MSR3}) is now equivalent to the following:
\begin{eqnarray}
&\mathrm{minimize}& \alpha_1 \\
&\textrm{subject to}& \sum_{j=1}^n \alpha_j \frac{b_j}{r_j}\left(e^{\frac{r_j}{b_j}d_{j-1}}-e^{\frac{r_j}{b_j}d_j}\right) = 1\label{PalphaNormolise}\\
&& \alpha_j e^{r_j d_{j}/b_j} = \alpha_{j-1} e^{r_{j-1}d_{j}/b_{j-1}}, \forall j = 2,\cdots, n \nonumber\\
&& \alpha_j > 0,\quad \forall j\in [n]\nonumber\\
&& B=d_1\ge d_2\ge\cdots\ge d_n\ge d_{n+1}=0 \nonumber
\end{eqnarray}
where (\ref{PalphaNormolise}) is computed directly from (\ref{problem:MSR2}b).

In Theorem~\ref{theorem:MSR}, if we know $(d_1,d_2,\cdots,d_{n+1})$, then we can see that $\alpha_j$ is proportional to $\alpha_1$. Therefore, we can get a constant $\Omega_j$ such that $\Omega_j\alpha_1=\int_{d_{j+1}}^{d_j}p_j(x)dx$ since the breakpoints are known. Finally we can get a constant $\Omega=\sum_{j=1}^{n}\Omega_j$ such that $\Omega\alpha_1=\sum_{j=1}^{n}(\int_{0}^{B}p_j(x)dx)$. Using the fact that $\sum_{j=1}^{n}(\int_{0}^{B}p_j(x)dx)=1$, we can easily solve $\alpha_1$, all the $\alpha_j$ and the whole problem.

Therefore, the computation of $\mathbf{p}$ reduces to computing $\{d_1, d_2,\cdots,d_{n+1}\}$. Notice that $d_1\equiv B, d_{n+1}\equiv 0$.

In this case, we treat this problem from another prospective. We first fix $\alpha_1$ to be 1. After that, without considering the constraint (\ref{PalphaNormolise}), we compute the optimal breakpoints $(d_1, d_2,\cdots,d_{n+1})$ to maximize $\sum_{j=1}^{n} (\int_{0}^{B} p_j(x)dx)$. Denote the optimal value of this problem as $\Omega$. We then normalize all the probability functions, i.e., reset all the $\alpha_j$ to be $\alpha_j/\Omega$. By Lemma~\ref{lemma:AlphaRelation}, we know the ratio $\lambda$ is proportional to $\alpha_1$, which is fixed at first and normalized at last. Hence, maximizing $\Omega$ is equivalent to minimizing $\lambda$. Notice that all the probability functions in the remainder of section~\ref{sec:computing} is unnormalized when $\alpha_1=1$.

In the following 2 sections~\ref{sec:compdn} and \ref{sec:compdj}, we show some intuitions and ideas of our algorithm about how to compute the breakpoints. In Section~\ref{sec:msralg}, we formally propose our algorithm and prove the optimality and complexity of our algorithm.

\subsubsection{Computing $d_n$}\label{sec:compdn}
To facilitate further calculations, we denote $P_{j}$ to be the probability sum of shop $j$ to shop $n$, i.e.,
$$P_{j}\triangleq\sum_{\tau=j}^{n} (\int_{0}^{B} p_\tau(x)\ud x)=\sum_{\tau=j}^{n} (\int_{0}^{d_j} p_\tau(x)\ud x)$$

Now we just need to maximize $P_1$ since by definition $P_1=\Omega$. To compute some breakpoint $d_j$, we assume that all the breakpoints $\{d_i:i\ne j\}$ are fixed. Since breakpoints $d_1,d_2,\cdots, d_{j-1}$ are fixed, parameters $\alpha_1,\alpha_2,\cdots,\alpha_{j-1}$ are constants. Therefore, $\sum_{\tau=1}^{j-2} (\int_{0}^{B}p_\tau(x) dx)$, part of the probability sum, is a constant and we just need to maximize the rest of the sum which is $P_{j-1}$.

First we consider how to compute $\arg\max_{d_n} P_{n-1}(d_n)$ when given $d_1,\cdots, d_{n-1}$, where
$$P_{n-1}(d_n)=\alpha_n \int_{0}^{d_n} e^\frac{r_nx}{b_n}dx+\alpha_{n-1}\int_{d_n}^{d_{n-1}} e^\frac{r_{n-1}x}{b_{n-1}}dx$$
Notice that $\alpha_{n-1}$ is a constant but $\alpha_n$ depends on $d_n$. From Lemma~\ref{lemma:AlphaRelation} we know that:
$$\alpha_n=\alpha_{n-1} b_{n-1} e^{(r_{n-1}/b_{n-1}-r_n/b_n)d_n}/b_n$$

The following lemma shows the concavity of $P_{n-1}(d_n)$:
\begin{lemma}
\label{lemma:dnconcave}
$P_{n-1}(d_n)$ is a strictly concave function.
\end{lemma}

Notice that $P'_{n-1}(d_n)>0$ when $d_n=0$. This implies that: if $d_n<d_{n-1}$, we must have $P'_{n-1}(d_n)=0$ since it is concave. Otherwise, $d_{n-1}=d_n$ which means $\forall x,p_{n-1}(x)=0$, i.e., shop $n-1$ does not exist. Thus we can delete shop $n-1$ and view shop $n-2$ as shop $n-1$. Similarly, if $d_n<d_{n-2}$, $P'_{n-2}(d_n)=0$; otherwise delete shop $n-2$ and treat shop $n-3$ as shop $n-1$. Repeat this procedure until we find some shop $k$, such that $d_n=d_{n-1}=\cdots=d_{k+1}<d_k$. Then $d_n$ should be the maximal point because of the concavity derived by Lemma~\ref{lemma:dnconcave}, i.e.,
$$d_n=\frac{b_n}{r_n}\ln(\frac{b_{k}r_n-b_nr_{k}}{b_n(r_n-r_{k})})$$
Notice that $d_n$ is always positive.

\subsubsection{Computing $d_j$}\label{sec:compdj}

Notice that $d_n$ is unrelated to $d_{n-1}$ if $d_n<d_{n-1}$. Therefore, we can work out all the breakpoints $d_j$ in descending order of the subscript of d. Here we show how to obtain $d_j$ after $d_n,d_{n-1},\cdots,d_{j+1}$.

If $j=n$, we just temporarily take
$$d_n=\frac{b_n}{r_n}\ln(\frac{b_{n-1}r_n-b_nr_{n-1}}{b_n(r_n-r_{n-1})})$$

If $j\ne n$, our target becomes $\arg\max_{d_j} P_{j-1}(d_j)$. According to the definition, we have
$$P_{j-1}(d_j)=\alpha_j (D_j+\int_{0}^{d_j} e^\frac{r_jx}{b_j}dx)+\alpha_{j-1}\int_{d_j}^{d_{j-1}} e^\frac{r_{j-1}x}{b_{j-1}}dx$$
where$$D_j\triangleq -\int_{0}^{d_{j+1}}e^{r_jx/b_j}dx+\sum_{\tau=j+1}^{n}\frac{\alpha_\tau}{\alpha_j}\int_{d_{\tau+1}}^{d_{\tau}}e^{r_{\tau}x/b_{\tau}}dx\ge 0$$

Notice that the breakpoints $d_n,d_{n-1},\cdots,d_{j+1}$ are fixed and we can compute $\alpha_\tau/\alpha_j$ by the following equations which is derived from Lemma~\ref{lemma:AlphaRelation}:
$$\alpha_\tau=\alpha_{\tau-1} b_{\tau-1} e^{(r_{\tau-1}/b_{\tau-1}-r_\tau/b_\tau)d_\tau}/b_\tau,\forall\tau\in[n]\backslash[j]$$
Therefore, $D_j$ is a constant.

It can be seen that we can compute $D_j$ recursively, i.e.,
$$D_j= \frac{\alpha_{j+1}}{\alpha_{j}}(D_{j+1}+\int_{0}^{d_{j+1}}e^{\frac{r_{j+1}}{b_{j+1}}x}\ud x -\int_{0}^{d_{j+1}}e^{\frac{r_j}{b_j}x}\ud x)$$
Also note that $\alpha_{j-1}$ is a constant but $\alpha_j$ depends on $d_j$:
$$\alpha_j=\alpha_{j-1} b_{j-1} e^{(r_{j-1}/b_{j-1}-r_j/b_j)d_j}/b_j$$
The following lemma shows that $P_{j-1}(d_j)$ is a quasi-concave function:

\begin{lemma}
\label{lemma:GConcave}
If $D_jr_j/b_j\ge 1$, we always have $P'_{j-1}(d_j)<0$; if $D_jr_j/b_j< 1$, $P''_{j-1}(d_j)<0$, i.e., $P_{j-1}(d_j)$ is strictly concave.
\end{lemma}

Similarly with the computation of $d_n$, if $D_jr_j/b_j>1$, then we always have $P'_{j-1}(d_j)<0$ and the optimal $d_j$ is $d_{j+1}$. Hence we delete shop $j$ and treat shop $j-1$ as shop $j$. Then we need to recompute $d_{j+1}$ and let $d_j=d_{j+1}$; otherwise it is concave and we temporarily get the maximal point:

$$d_j=\frac{b_j}{r_j}\ln(\frac{(b_{j-1}r_j-b_jr_{j-1})(1-D_jr_j/b_j)}{b_j(r_j-r_{j-1})})$$

Here if the temporary $d_j$ is no larger than $d_{j+1}$, it means that the optimal solution is $d_{j+1}$ because of the constraints $d_{j+1}\le d_j\le d_{j-1}$. So we have $d_j=d_{j+1}$ which means that $\forall x,p_j(x)=0$. Therefore, we delete shop $j$ and treat shop $j-1$ as shop $j$. Then recompute $d_{j+1}$ and temporarily skip $d_j$. At last we set $d_j=d_{j+1}$.

\subsubsection{A Linear Time Algorithm}\label{sec:msralg}
Now we are ready to show our algorithm for computing the optimal strategy of the consumer.
\begin{theorem}
\label{theorem:MSRalg}
There is an algorithm for computing the unique optimal strategy of the consumer. The time and space complexity of the algorithm are linear.
\end{theorem}

We first show how to construct our algorithm, and analyze the correctness and the complexity of our algorithm later.

Since delete operations may be executed frequently in the algorithm, we use a linked list to store the shop info. Each shop is an element in this linked list and the shop index decreases when we traverse from the head to the tail. So the head is shop $n$ and the tail is shop $1$. Considering that the shops appear in the form of linked list in the algorithm, we rewrite some equations we may use in the algorithm:
\begin{eqnarray}
\label{CD}
D_j&=& \frac{\alpha_{prev[j]}}{\alpha_{j}}(D_{prev[j]}+\int_{0}^{d_{prev[j]}}\exp(\frac{r_{prev[j]}x}{b_{prev[j]}})dx)\nonumber\\ &&-\int_{0}^{d_{prev[j]}}\exp(\frac{r_jx}{b_j})dx \nonumber \\
&=&\frac{\alpha_{prev[j]}}{\alpha_{j}}D_{prev[j]}-\frac{b_j}{r_j}(\exp(\frac{r_j d_{prev[j]}}{b_j})-1) \nonumber\\
&&+\frac{\alpha_{prev[j]}b_{prev[j]}}{\alpha_{j}r_{prev[j]}}(\exp(\frac{r_{prev[j]}d_{prev[j]}}{b_{prev[j]}})-1)
\end{eqnarray}
Here $\frac{\alpha_{prev[j]}}{\alpha_{j}}$ is represented as follow:
$$\frac{\alpha_{prev[j]}}{\alpha_{j}} =\frac{b_j}{b_{prev[j]}}e^{(\frac{r_j}{b_j} -\frac{r_{prev[j]}}{b_{prev[j]}})d_{prev[j]}}$$
\begin{eqnarray}
\label{BP}
d_j=\frac{b_j}{r_j} \ln(\frac{(b_{next[j]}r_j-b_jr_{next[j]})(1-D_jr_j/b_j)}{b_j(r_j-r_{next[j]})})
\end{eqnarray}
Here is the pseudocode of our algorithm.:
\begin{algorithm}[htb]
\caption{MSR Algorithm}
\label{alg:CompNoEF}
\begin{algorithmic}[1]
\STATE $D_n\leftarrow 0$;
\FOR {$j\leftarrow 1 \text{ to } n$}
\STATE $next[j]\leftarrow j-1$;
\STATE $prev[j]\leftarrow j+1$;
\ENDFOR
\FOR {$j\leftarrow n \text{ to } 2$}
\STATE $ComputingBP(j)$;
\ENDFOR
\FOR {$j\leftarrow n \text{ to } 2$}
\IF {$d_j\ne$"decide later"}
\IF {$d_j>B$}
\STATE $d_j\leftarrow B$;
\ENDIF
\ELSE
\STATE $d_j\leftarrow d_{j+1}$;
\ENDIF
\ENDFOR
\end{algorithmic}
\end{algorithm}
\begin{algorithm}[htb]
\caption{Function $ComputingBP(j)$}
\begin{algorithmic}[1]
\IF {$j\ne n$}
\STATE Update $D_j$ according to (\ref{CD});
\ENDIF
\IF {$D_j\ge b_j/r_j$}
\STATE $d_j\leftarrow$"decide later";
\STATE $next[prev[j]]\leftarrow next[j]$;
\STATE $prev[next[j]]\leftarrow prev[j]$;
\STATE $ComputingBP(prev[j])$;
\ELSE
\STATE Compute $d_j$ according to (\ref{BP});
\IF {$d_j\le d_{j+1}$}
\STATE $d_j\leftarrow$"decide later";
\STATE $next[prev[j]]\leftarrow next[j]$;
\STATE $prev[next[j]]\leftarrow prev[j]$;
\STATE $ComputingBP(prev[j])$;
\ENDIF
\ENDIF
\end{algorithmic}
\end{algorithm}

Though we may revise those breakpoints for many times when running the algorithm, it will still lead to the exact optimal solution at the end. Since the feasible solution of $(d_2,d_3,\cdots,d_n)$ is convex and functions $P_{j-1}(\cdot)$ are always concave, we have the following properties for the optimal solution:

 If $d_{j-1}>d_j$, $P'_{j-1}(d_j)\le 0$;
 if $d_{j+1}<d_j$, $P'_{j-1}(d_j)\ge 0$.

So in our computation method, we delete a shop when and only when the shop should be deleted in the optimal solution. Notice that Line $5,6,7$ and Line $15,16,17$ are what we actually do when we say we delete shop $j$. We say a shop is \emph{alive} if it has not been deleted. Based on the following lemma, we rigorously prove the correctness and complexity of this algorithm.

\begin{lemma}
After an invocation of $ComputingBP(j)$ is completed, the temporary breakpoints, whose indexes are less than or equal to $j$, $\mathbf{td}=(td_n,td_{n-1},\cdots,td_j)$ are identical with the optimal solution $\mathbf{d^*}=(d_n^*,d_{n-1}^*,\cdots,d_j^*)$ if $d_j^*<d_{next[j]}^*$. And all the deletions are correct, i.e., once a shop $j$ is deleted in the algorithm, $d_j^*$ must be equal to $d_{j+1}^*$.
\label{lemma:invoc}
\end{lemma}

Here $\mathbf{td}=(td_n,td_{n-1},\cdots,td_j)$ are the temporary values of $d_n,d_{n-1},\cdots,d_j$ just after this invocation, $d_n^*,d_{n-1}^*,\cdots,d_2^*$ are the optimal breakpoints, and $next[\cdot]$ and $prev[\cdot]$ denote the current state of the linked list, not the eventual result.


\begin{proof} (Theorem~\ref{theorem:MSRalg})
We first show the correctness of the algorithm. According to Lemma~\ref{lemma:invoc}, we know that all the deletions are correct. Also, we know that $\forall j_1,j_2$ such that $1<j_1<j_2$, and that shop $j_1$ and shop $j_2$ are alive, $td_{j_1}>td_{j_2}$ when the algorithm terminates. There are 2 cases:

Case 1: $B=d_1^*>d_{prev[1]}^*$, the final solution is the unique optimal solution by Lemma~\ref{lemma:invoc}.

Case 2: $B=d_1^*=d_{prev[1]}^*$. Similar to the proof of Case 1 in Lemma~\ref{lemma:invoc}, the solution of the alive breakpoints $\mathbf{td^*}=(td_n^*,td_{next[n]}^*,\cdots,td_{prev[1]}^*,d_1^*)$, satisfying that $\forall \tau\in[n], td_\tau^*=\min\{td_\tau, d_{1}^*\}$, is the unique optimal solution.

Next, we analyze the complexity of the algorithm. Obviously, the space complexity is $O(n)$. For the time complexity, we just need to prove that this algorithm invokes the function $ComputingBP$ for $O(n)$ times. The main function invokes $ComputingBP$ for $O(n)$ times. And notice that a shop is deleted once $ComputingBP$ is invoked recursively by $ComputingBP$. The algorithm can delete $O(n)$ shops at most, therefore the total number of the invocations is $O(n)$.
\end{proof}

\subsection{Optimal Strategy of Nature}
Now we briefly show the optimal strategy $\mathbf{q^*}$ of nature in the following theorem:

\begin{theorem}
\label{theorem:nature}
Suppose that the probability for nature of choosing $y=B$ (actually $y=+\infty$) is $q_B^*$. The optimal strategy $q^*$ of nature satisfies the following properties:
\begin{enumerate}
\item[(a)] $\forall j>1$, we have $\int_{d_j^{*-}}^{d_j^{*+}}q^*(y)\ud y=0$.
\item[(b)] $q^*(y)=\beta_j y e^{-\frac{r_j}{b_j}y},\quad\quad y\in(d_{j+1}^*,d_j^*)$.
\item[(c)] $\beta_1=(q_B r_1/B b_1)e^{\frac{r_1}{b_1}B}$.
\item[(d)] $\frac{b_j}{r_j}\beta_j e^{-\frac{r_j}{b_j}d_j^*}=\frac{b_{j-1}}{r_{j-1}} \beta_{j-1} e^{-\frac{r_{j-1}}{b_{j-1}}d_j^*},\quad \quad j=2,3,\ldots,n$.
\end{enumerate}
\end{theorem}

Recall that $\mathbf{d^*}$ is computed by Algorithm~\ref{alg:CompNoEF}. Thus $\forall j\in[n]$, we can compute $\beta_j/q_B$ and we can obtain $1/q_B$ by computing the sum $\sum_{\tau=1}^{n}\int_{d_\tau}^{d_{\tau+1}}q(y)/q_B \ud y$. Since the sum of probability is 1, i.e., $\int_{0}^{B} q(y) \ud y=1$, it is not hard to work out $\mathbf{q}^*$ after normalization.

\subsection{Including Switching Cost}

In this subsection, we consider an extension, the \emph{multi-shop ski rental with switching cost problem}(also \texttt{MSR-S}), in which we allow the consumer switches from one shop to another, but with some extra fee to be paid. At any time, suppose the consumer chooses to switch from shop $i$ to shop $j$, she has to pay for an extra switching cost $c_{ij}\ge 0$. If there exists some $i \neq j$, such that $c_{ij} = +\infty$, then the consumer cannot \emph{directly} switch from shop $i$ to shop $j$.
Consider the following 2 cases:
\begin{itemize}
\item If the consumer is allowed to switch from shop to shop freely, i.e. the switching cost is always 0, she will optimally rent at the shop with the lowest renting price and buy at the shop with the lowest buying price. All that she concerns is when to buy the skies. Thus, this problem (\texttt{MSR-S}) is reduced to the basic ski rental problem (\texttt{SR}).
\item If the switching cost is always $+\infty$, she will never switch to another shop and the \texttt{MSR-S} becomes \texttt{MSR}.
\end{itemize}
We will prove that the consumer never switches between shops even in  \texttt{MSR-S} later.

The settings of \texttt{MSR-S} can be viewed as a directed graph $G=(V,A)$, where $V = \{1,\cdots,n\}$, and $A = \{(i,j): c_{ij}<+\infty\}$. Each arc $(i,j)\in A$ has a cost $c_{ij}$. We define a path $\mathbf{p} \subseteq G$ as a sequence of arcs. Define the cost of $\mathbf{p}$ as the summation of the costs of all arcs on $\mathbf{p}$. Note that if the consumer is allowed to switch from shop $i$ to shop $j$ $(i\neq j)$, there must be a path $\mathbf{p}$ which starts at $i$ and ends at $j$.

It is clear that if the consumer decides to switch from shop $i$ to shop $j$ ($(i,j) \in A$) at any time, she will choose the shortest path from $i$ to $j$ in the graph $G$. Denote the cost of the shortest path from $i$ to $j$ by $c_{ij}^*$. We obtain that $c_{ij}^* \leq c_{ij}$, for any $(i,j) \in A$. Moreover, for any different $i, j , k\in V$ such that $(i,j), (j,k), (i,k) \in A$, we have:
\begin{equation} \label{relation:MSR-S}
c_{ik}^* \leq c_{ij}^*+c_{jk}^*
\end{equation}

Comparing to \texttt{MSR}, \texttt{MSR-S} has a much richer action set for the consumer, where the consumer is able to choose where to rent and for how long to rent at that shop.

Although we allow the consumer to switch from shop to shop as many times as she wants, the following lemma shows that the consumer will never choose to switch to another shop and continue renting, i.e. the only moment that the consumer will switch is exactly when she buys the skis.

\begin{lemma} \label{lemma:MSR-S}
Any strategy which includes switching from one shop to another shop and continuely renting the skis is dominated.
\end{lemma}

This lemma significantly reduces the action set of the consumer to the same one as \texttt{MSR}. Thus, if the consumer considers to switch, she must switch to another shop at the buying time. Further, she chooses a shop such that the sum of the buying cost and the switching cost (if any) is minimized. Once the consumer decides to switch for buying, she will switch at most once, since the buying cost only increases otherwise. Therefore, for any strategy, suppose $s$ is the shop in which the consumer rents the skis right before the buying time, we define the buying price of $s$ as follows:
\begin{displaymath}
b_s' = \min\{b_s, b_{j}+c_{sj}, \forall j \neq s\}
\end{displaymath}
Observe that once $s$ is settled, $b_s'$ is settled. As a result, \texttt{MSR-S} is reduced to \texttt{MSR}, in which for any shop $j$,  the rent is still $r_j$ per unit time while the buying price $b'_j$ is $\min \{b_j, \min_{i\neq j}\{b_{i}+c_{ji}\}\}$.
\section{Ski rental with ENTRY FEE}

In this section, we discuss another extension of \texttt{MSR}, the \emph{multi-shop ski rental with entry fee included problem} (\texttt{MSR-E}). In this problem, all the settings are the same to those of \texttt{MSR}, except that each shop has an entry fee. Once the consumer enters a shop, she pays for the entry fee of this shop and cannot switch to another shop. Our goal is to minimize the worst case competitive ratio. Notice that \texttt{MSR} can be viewed as a special case of \texttt{MSR-E} in which the entry fee of each shop is zero.

We introduce this problem not only as an extension of \texttt{MSR}, but more importantly, as a necessary step to solve a more general extension, the (\texttt{MSR-ES}) problem in next section. We will show that \texttt{MSR-ES} can be converted into \texttt{MSR-E} with minor modifications.

\subsection{Single Shop Ski Rental with Entry Fee}
We start by briefly introducing the special case of \texttt{MSR-E} when $n=1$. The entry fee, renting price and buying price are supposed to be $a\ge 0$, $r>0$ and $b>0$. Without loss of generality, we assume that $r=1$.

It can be verified that
\vspace{-1mm}
\begin{enumerate}
\item[(i)] Using dominance, the buying time of the consumer $x\in[0,b]$, and the stopping time chosen by nature $y\in(0,b]$.
\vspace{-2mm}
\item[(ii)] For all $ y\in(0,b]$, the ratio is a constant if the consumer chooses the optimal mixed strategy.
    \vspace{-2mm}
\item[(iii)] No probability mass appears in $(0,b]$.
\end{enumerate}

By calculation, we obtain the following optimal mixed strategy:
\vspace{-2mm}
\begin{itemize}
\item The probability that the consumer buys at time $x=0$ is
$p_0=a/((a+b)e-b)$.
\vspace{-2mm}
\item The probability density function that the consumer buys at time $x\in(0,b]$ is
$p(x)=\frac{\exp(x/b)}{b(e-\frac{b}{a+b})}$.
\vspace{-2mm}
\item The competitive ratio is $\frac{e}{e-\frac{b}{a+b}}$.
\end{itemize}

Note that the biggest difference from \text{MSR} is that $p_0$ may be  probability mass, which means that the consumer may have non-zero probability to buy at the initial time.
\subsection{Analysis of MSR-E}

In this problem, assume that there are $n$ shops in total. For any shop $j\in [n]$, the entry fee, renting price and buying price of shop $j$ are $a_j\ge 0$, $r_j>0$ and $b_j>0$, respectively. Similar to the procedures in \texttt{MSR}, we use a tuple $(j,x)$ in which $j\in[n],x\in[0,+\infty)\cup\{+\infty\}$ to denote an action for the consumer and a number $y\in(0,+\infty)\cup\{+\infty\}$ for nature.

Without loss of generality, in this problem, we assume that
\begin{itemize}
\item $r_1\le r_2\le\cdots\le r_n$;
\item $\forall i,j,~a_i<a_j+b_j$;
\item $\forall i<j$, $a_i>a_j$ or $a_i+b_i>a_j+b_j$.
\end{itemize}
The second condition is because that shop $i$ is dominated by shop $j$ if $a_i\ge a_j+b_j$. For the third condition, we know $r_i\le r_j$ since $i<j$. So shop $j$ is dominated by shop $i$ if we also have $a_i\le a_j$ and $a_i+b_i\le a_j+b_j$.

Denote $B$ as follows:
\begin{eqnarray*}
&\mathrm{minmize}& B\\
&\mathrm{subject~to}& \forall i,~a_i+Br_i\ge\min_j(a_j+b_j)
\end{eqnarray*}
Similar to Lemma~\ref{lemma:MSRstspace} in \texttt{MSR}, we reduce the action sets for both players by the following lemma:
\begin{lemma}
\label{lemma:EFstspace}
For nature, any action $y \in [B,+\infty)$ is dominated. For the consumer, any action $(j,x)$  is dominated, in which $x \in (B,+\infty)\cup\{+\infty\}, j \in [n]$.
\end{lemma}

Similar to \texttt{MSR}, the consumer's action set is reduced to buying time $x\in [0,B]$, and nature's action set is reduced to $\{y\in(0,B]\}$.

The strategy spaces for both the consumer and the nature are identical to those in \texttt{MSR}. Similarly, we use $\mathbf{p}$ to denote a mixed strategy of the consumer and $\mathbf{p^*}$ to denote the optimal mixed strategy. If the consumer chooses mixed strategy $\mathbf{p}$ and nature chooses $y$, we denote the cost function as follows:
\begin{eqnarray}
C(\mathbf{p},y) &=& \sum_{j\in[n]}\bigg(\int_0^y (a_j+r_j x + b_j) p_j(x) \ud x \nonumber\\
&& +  \int_y^B (a_j+r_j y) p_j(x) \ud x\bigg)
\end{eqnarray}
We define $\mathrm{OPT}(y)$ as the offline optimal strategy when nature chooses the action $y$, i.e.
$$\mathrm{OPT}(y)=\min_j\{a_j+r_jy\},\quad y\in(0,B]$$
By \cite{decomputational}, we can compute the function $\mathrm{OPT}(y)$ in linear time. The objective of the consumer is $\min_\mathbf{p}\max_y\frac{C(\mathbf{p},y)}{\mathrm{OPT}(y)}$.

Similar to \texttt{MSR}, we give the following lemmas:
\begin{lemma}\label{lemma:MSR-E constant}
For the optimal strategy $p^*$ of the consumer, $\frac{C(\mathbf{p^*},y)}{\mathrm{OPT}(y)}$ is a constant for any $ y \in (0,B],$.
\end{lemma}

\begin{lemma} \label{lemma:MSR-E finite}
$\forall x \in (0,B], p^*_j(x) < +\infty$.
\end{lemma}

The problem is formalized as follows:
\begin{align}
\textrm{minimize}& \quad\quad\quad\lambda \label{problem:MSREF1}\\
\textrm{subject to}& \quad \frac{C(\mathbf{p}, y)}{\mathrm{OPT}(y)} = \lambda,\quad\forall y \in (0,B]\tag{\ref{problem:MSREF1}a} \\
& \quad \sum_{j=1}^n \int_{0}^B p_j(x) \ud x = 1 \tag{\ref{problem:MSREF1}b}\\
& \quad p_j(x) \geq 0,\quad \forall x \in [0,B] \tag{\ref{problem:MSREF1}c}
\end{align}
Note that there may be probability mass at $x=0$. For Problem~\ref{problem:MSREF1}, we find that the optimal mixed strategy $\mathbf{p^*}$ of the consumer is also segmented.

\begin{lemma} \label{lemma:moveP MSR-E}
In the optimal mixed strategy $\mathbf{p}^*$, there exists $n+1$ breakpoints $B=d_{1}\ge d_{2}\ge\cdots\ge d_{n+1} = 0$, which partition $[0,B]$ into $n$ sub-intervals, such that $\forall j=1,\cdots,n$, $\forall x \in (d_{j+1},d_{j})$, $p_{j}^*(x) > 0$ and $p_{i}^*(x) = 0$ for any $i \neq j$.
\end{lemma}

\noindent \textbf{Remark:} Though we prove the form of the optimal solution to \texttt{MSR-E}, computing the analytical solution is very challenging. The point mass at $x=0$ makes this problem much more difficult than \texttt{MSR}, because one needs to guarantee the nonnegativity of this probability mass. Moreover, the non differentiable points of $\mathrm{OPT}(y)$, which we call the offline breakpoints, complicate the probability density function form in each segment $(d_{j+1},d_j)$. In fact, we can obtain the exact analytic optimal solution when $n=2$.

\section{ski rental with ENTRY FEE AND SWITCHING}

Now we introduce the last extension of \texttt{MSR}, the \emph{entry fee included, switching allowed problem} (\texttt{MSR-ES}). All the settings are identical to those of \texttt{MSR-E}, except that the consumer is allowed to switch at any time. When a consumer enters or switches to a shop, she pays the entry fee of the shop. For instance, if a consumer enters shop 1 at first, then switches from shop 1 to shop 2 and returns to shop 1 at last, she pays the entry fee of shop 1 twice and the entry fee of shop 2 once.

Similar to \texttt{MSR-E}, there exist $n$ shops. The entry fee, renting price, buying price of shop $j$ are denoted by $a_j\ge 0, r_j>0, b_j>0$, respectively. Without loss of generality, we assume that
\begin{itemize}
\item $r_1\le r_2\le\cdots\le r_n$;
\item $\forall i,j,~a_i<a_j+b_j$;
\item $\forall i<j$, $a_i>a_j$ or $a_i+b_i>a_j+b_j$.
\end{itemize}
As in the \texttt{MSR-S} case, if there exists $ i,j\in [n]$ such that $b_i> a_j+b_j$, then instead of buying in shop $i$, the consumer will switch from shop $i$ to shop $j$ to buy skis  \footnote{This phenomenon is called ``switching for buying''.}. This is equivalent to setting $b_i$ to be $\min_{j\neq i}\{ a_j+b_j\}$. Therefore, without loss of generality, we assume that
\begin{itemize}
\item $\forall i,j,~b_i\le a_j+b_j$.
\end{itemize}
In the remainder of this section, we first define the action set and formulate our problem. Then, we show that the strategy space can be reduced by Lemma~\ref{lemma:EAdominate}, \ref{lemma:reduceSpace} and \ref{lemma:EAconstant}. In Lemma~\ref{lemma:EAf2p}, we show that for each switching operation, the operations of the consumer before or after the switching is not important. The only things we care about are when the switching happens, and which shops the switching operation relates to. Thus, we can construct a virtual shop for each switching operation, and (nearly) reduce this \texttt{MSR-ES} problem to \texttt{MSR-E}. Finally, we show that \texttt{MSR-ES} has the similar nice properties as \texttt{MSR-E} which we have known in Lemma~\ref{lemma:MSR-E finite} and \ref{lemma:moveP MSR-E}.
\subsection{Notations and Analysis of MSR-ES}
\subsubsection{Reduced Strategy Space}
The action set for nature is $\{y>0\}$, defined as before. To represent the action set formally, we firstly introduce the operation tuple $\sigma=(i,j,x)$ to denote the switching operation of switching from shop $i$ to shop $j$ at time $x$. For special cases, $(0,j,0)$ denotes the entering operation that the consumer enters shop $j$ at the very beginning; $(i,0,x)$ denotes the buying operation that the consumer buys at shop $i$ at time $x$. An operation tuple can also be represented as $(j,x)$ for short, denoting the switching operation to shop $j$ at time $x$ if $j>0$, the entering operation if $x=0$, and the buying operation at time $x$ if $j=0$.

Then, an action $\psi$ is expressed as a sequence (may be infinite) of the operation tuples:
$$\psi=\{(j_0,x_0),(j_1,x_1),(j_2,x_2),\cdots\}$$
satisfying that
\begin{itemize}
\item $0=x_0\le x_1\le x_2\le\cdots$;
\item if there exists $x\ge 0$ such that $(0,x)\in \psi$, it is the last element in $\psi$.
\end{itemize}
or the full form with the same constraints:
$$\psi=\{(0,j_0,x_0),(j_0,j_1,x_1),(j_1,j_2,x_2),\cdots\}$$
Similar to other extensions, we reduce the action set. In this model, the definition of $B$ is the same as those of \texttt{MSR-E}:
\begin{eqnarray*}
&\mathrm{minimize}& B\\
&\mathrm{subject~to}& \forall i,~a_i+Br_i\ge\min_j(a_j+b_j)
\end{eqnarray*}
and we give the following lemma:
\begin{lemma}\label{lemma:EAdominate}
From the perspective of nature, any strategy $y\in[B,+\infty)$ is dominated. While for the consumer, any strategy in which the buying time $x\in (B,+\infty)\cup \{+\infty\}$ is dominated.
\end{lemma}

Similar to \texttt{MSR}, we reduce the consumer's buying time to the interval $[0,B]$, and nature's action set to $\{y\in(0,B]\}$.

The following lemma shows that a consumer may switch from shop $i$ to shop $j$ for renting, only when $r_i>r_j$ and $a_i<a_j$.
\begin{lemma} \label{lemma:reduceSpace}
If a strategy of the consumer: $$\psi = \{(0,j_0,x_0),(j_0,j_1,x_1),\cdots,(j_{|\psi|-2},0,x_{|\psi|-1})\}$$ satisfies any of the following conditions, then it is dominated.
\begin{itemize}
\item $\exists 0<\tau<|\psi|-1$ such that $x_{\tau-1}=x_\tau$;
\item $\exists (i,j,x)\in \psi$ such that $r_i\le r_j$ and $(j,0,x)\notin\psi$;
\item $\exists (i,j,x)\in \psi$ such that $a_i\ge a_j$ and $(j,0,x)\notin\psi$.
\end{itemize}
\end{lemma}

Here we give some intuitions. In these three cases, we can construct a new action $\psi'$ by deleting one specified operation from $\psi$, and show that $\psi$ is dominated by $\psi'$.

This lemma rules out a huge amount of dominated strategies from our action set and allows us to define the operation set:
\begin{eqnarray*}
\Sigma&\triangleq&\bigg\{\sigma=(i,j,x):i,j\in[n], \ r_i>r_j,\ a_i<a_j, \\
&&x\in(0,B]\bigg\}\bigcup\bigg\{(0,j,0):j\in[n]\bigg\}\\&&\bigcup\bigg\{(j,0,x):j\in[n],x\in[0,B]\bigg\}
\end{eqnarray*}
Thus, we only need to consider such an action set:
\begin{eqnarray*}
\Psi_c\triangleq\bigg\{\psi&=&\{(0,j_0,x_0),(j_0,j_1,x_1),\cdots,(j_{|\psi|-2},0,x_{|\psi|-1})\}\\
&:&0=x_0<x_1<\cdots<x_{|\psi|-2}\le x_{|\psi|-1}\le B,\\
&&r_{j_0}>r_{j_1}>\cdots>r_{j_{|\psi|-2}}~,\\
&&a_{j_0}<a_{j_1}<\cdots<a_{j_{|\psi|-2}}
\bigg\}
\end{eqnarray*}
Since $r_{j_0}>r_{j_1}>\cdots>r_{j_{k}}$, we get $j_0>j_1>\cdots>{j_{k}}$ and $2\leq |\psi|\le n+1$.
\subsubsection{Mathematical Expression of the Cost, Ratio and the optimization problem}
For nature's action $y\in(0,B]$ and the consumer's action $\psi=\{(j_0,x_0),\cdots,(0,x_{|\psi|-1})\}\in\Psi_c$, we define the cost $c(\psi,y)$ as follows:
\begin{align*}
c(\psi,y)\triangleq
\begin{cases}
\sum_{\tau=0}^{k-1} [a_{j_\tau}+r_{j_\tau} (x_{\tau+1}-x_{\tau})]+r_{j_{k}}(y-x_k),\\
\quad\quad\quad\quad\text{ if }\exists 0<k<|\psi|, x_{k-1}\le y<x_k;\\
\sum_{\tau=0}^{|\psi|-2} (a_{j_\tau}+r_{j_\tau} (x_{\tau+1}-x_{\tau}))+b_{j_{|\psi|-2}},\\
\quad\quad\quad\quad\text{ if }y\ge x_{|\psi|-1}.
\end{cases}
\end{align*}

For any action $\psi=\{(j_0,x_0),\cdots,(0,x_{|\psi|-1})\}$, we use $\mathbf{s}(\psi)$ to denote the order of the operations:
$$\mathbf{s}(\psi)\triangleq \{(0,j_0),(j_0,j_1),\cdots,(j_{|\psi|-2},0)\}$$
or the short form:
$$\mathbf{s}(\psi)\triangleq \{j_0,j_1,\cdots,j_{|\psi|-2},0\}$$
Further, we define $\mathcal{S}$ as the collection $\mathbf{s}(\Psi_c)$ as follows:
\begin{eqnarray*}
\mathcal{S}\triangleq\{\mathbf{s}&=&\{j_0,j_1,\cdots,j_k,0\}~:~k\ge 0,\\
&&r_{j_0}>r_{j_1}>\cdots>r_{j_{k}}~,~a_{j_0}<a_{j_1}<\cdots<a_{j_{k}}
\}
\end{eqnarray*}

Note that $j_0,j_1,\cdots,j_k\in[n]$ and $\{0\}\notin \mathcal{S}$, so the amount of elements in $\mathcal{S}$ is upper bounded by $|\mathcal{S}|\le 2^n-1$.

\noindent We group all the actions in $\Psi_c$ whose $\mathbf{s}(\psi)$ are identical. Thus, we partition $\Psi_c$ into $|\mathcal{S}|$ subsets.

\noindent For any action $\psi$, let $\mathbf{x}(\psi)$ denote the sequence of the operation time, defined as follows:
$$\mathbf{x}(\psi)\triangleq(x_1,\cdots,x_{|\psi|-1})$$

For each operation order $\mathbf{s}\in \mathcal{S}$, we define $\mathcal{X}_\mathbf{s}$ as the collection $\{\mathbf{x}(\psi):\mathbf{s}(\psi)=\mathbf{s}\}$, i.e.,
$$
\mathcal{X}_\mathbf{s}\triangleq\{\mathbf{x}=(x_1,x_2,\cdots,x_{|\mathbf{s}|-1}):
0<x_{1}<\cdots<x_{|\mathbf{s}|-1}\le B\}
$$

We observe that any $ \mathbf{s}\in \mathcal{S}$ and $ \mathbf{x}\in\mathcal{X}_\mathbf{s}$ can be combined to a unique action $\psi(\mathbf{s},\mathbf{x})$. Further, we can use $c_{\mathbf{s}}(\mathbf{x},y)$ and $c(\psi({\mathbf{s}},\mathbf{x}),y)$ interchangeably.

For each $\mathbf{s}\in\mathcal{S}$, we define the probability density function $f_{\mathbf{s}}:\mathcal{X}_\mathbf{s}\rightarrow[0,+\infty)\cup \{+\infty\}$
\footnote{ $f_{\mathbf{s}}(\mathbf{x})=+\infty$ represents probability mass on $\mathbf{x}$. }
, which satisfies $$\sum_{\mathbf{s}\in\mathcal{S}}\idotsint\limits_{\mathbf{x}\in\mathcal{X}_\mathbf{s}} f_{\mathbf{s}}(\mathbf{x})\ud\mathbf{x}=1$$

Let $\mathbf{f}=\{f_\mathbf{s}:\mathbf{s}\in\mathcal{S}\}$ denote a mixed strategy for the consumer.
Given a mixed strategy $\mathbf{f}$ of the consumer and nature's choice $y$, the expected competitive ratio is defined as follows:
\begin{equation}
R(\mathbf f, y)\triangleq\frac{C(\mathbf f,y)}{\mathrm{OPT}(y)}
\end{equation}
where $\mathrm{OPT}(y) = \min_{j\in[n]}\{a_j+r_j y\}$, and
\begin{equation}
\label{def:jointExpCost}
C(\mathbf{f}, y) \triangleq \sum_{\mathbf{s}\in\mathcal{S}}\idotsint\limits_{\mathbf{x}\in\mathcal{X}_\mathbf{s}}c_{\mathbf{s}}(\mathbf{x},y) f_{\mathbf{s}}(\mathbf{x})\ud\mathbf{x}
\end{equation}

\noindent The objective of the consumer is $\min_{\mathbf{f}} \max_{y} R(\mathbf{f}, y)$. The following lemma proves that $\forall y\in(0,B]$, $R(\mathbf{f^*},y)$ is a constant in which $\mathbf{f^*}$ is an optimal mixed strategy.
\begin{lemma}\label{lemma:EAconstant}
If $\mathbf{f^*}$ is an optimal solution of the problem $\arg\min_{\mathbf{f}} \max_{y} R(\mathbf{f}, y)$, then there exists a constant $\lambda$ such that $\forall y\in(0,B]$, $R(\mathbf{f^*}, y)=\lambda$.
\end{lemma}

The formalized optimization problem is as follows:
\begin{align}
\mathrm{minimize}& ~~\lambda \label{problem:EA1}\\
\mathrm{subject~ to}& ~~\frac{C(\mathbf f,y)}{\mathrm{OPT}(y)} = \lambda,\forall y \in (0,B]\tag{\ref{problem:EA1}a}\\
& ~~\sum_{\mathbf{s}\in\mathcal{S}}\idotsint\limits_{\mathbf{x}\in\mathcal{X}_\mathbf{s}} f_{\mathbf{s}}(\mathbf{x})\ud\mathbf{x}=1 \tag{\ref{problem:EA1}b}\\
& ~~f_{\mathbf s}(\mathbf x) \geq 0,\forall \mathbf{s}\in\mathcal{S}\tag{\ref{problem:EA1}c}
\end{align}

\subsection{Reduction to MSR-E}
For a mixed strategy $\mathbf{f}$, we define the probability density function of an operation $\sigma=(i,j,x)$ as follows:
\begin{equation} \label{def:pdfForEvent}
\mathbf{p}^{(\mathbf{f})}(\sigma) \triangleq \sum_{\mathbf{s}\in\mathcal{S}:(i,j)\in\mathbf{s}}~\idotsint\limits_{\mathbf{x}_{-\{x\}}:\sigma\in\psi(\mathbf{s},\mathbf{x})} f_{\mathbf{s}}(\mathbf{x})\ud(\mathbf{x}_{- \{x\}})
\end{equation}
where $\mathbf{x}_{-\{x\}}$ is the vector $\mathbf{x}$ in which the element corresponding to $x$ is eliminated. Here $p_{(i,j)}^{(\mathbf{f})}(x)$ can also be viewed as a marginal probability density function.
Also the p.d.f of an operation can be expressed in another form:
$$p_{(i,j)}^{(\mathbf{f})}(x)\triangleq \mathbf{p}^{(\mathbf{f})}((i,j,x))$$

Then we give the following lemma:
\begin{lemma}\label{lemma:EAf2p}
For any 2 mixed strategies $\mathbf{f_1},\mathbf{f_2}$ for the consumer, we have $C(\mathbf{f_1},y)=C(\mathbf{f_2},y)$ for all $y\in(0,B]$ if $\mathbf{p}^{(\mathbf{f_1})}(\sigma)=\mathbf{p}^{(\mathbf{f_2})}(\sigma)$ for all $\sigma\in\Sigma$, i.e., for a mixed strategy $\mathbf{f}$, we only care about its marginal ${\mathbf{p}}^{(\mathbf{f})}(\sigma)$.
\end{lemma}

\begin{figure*}[htb]
\subfigure{
\begin{minipage}[b]{0.32\linewidth}
\centering
\includegraphics[width=1\textwidth]{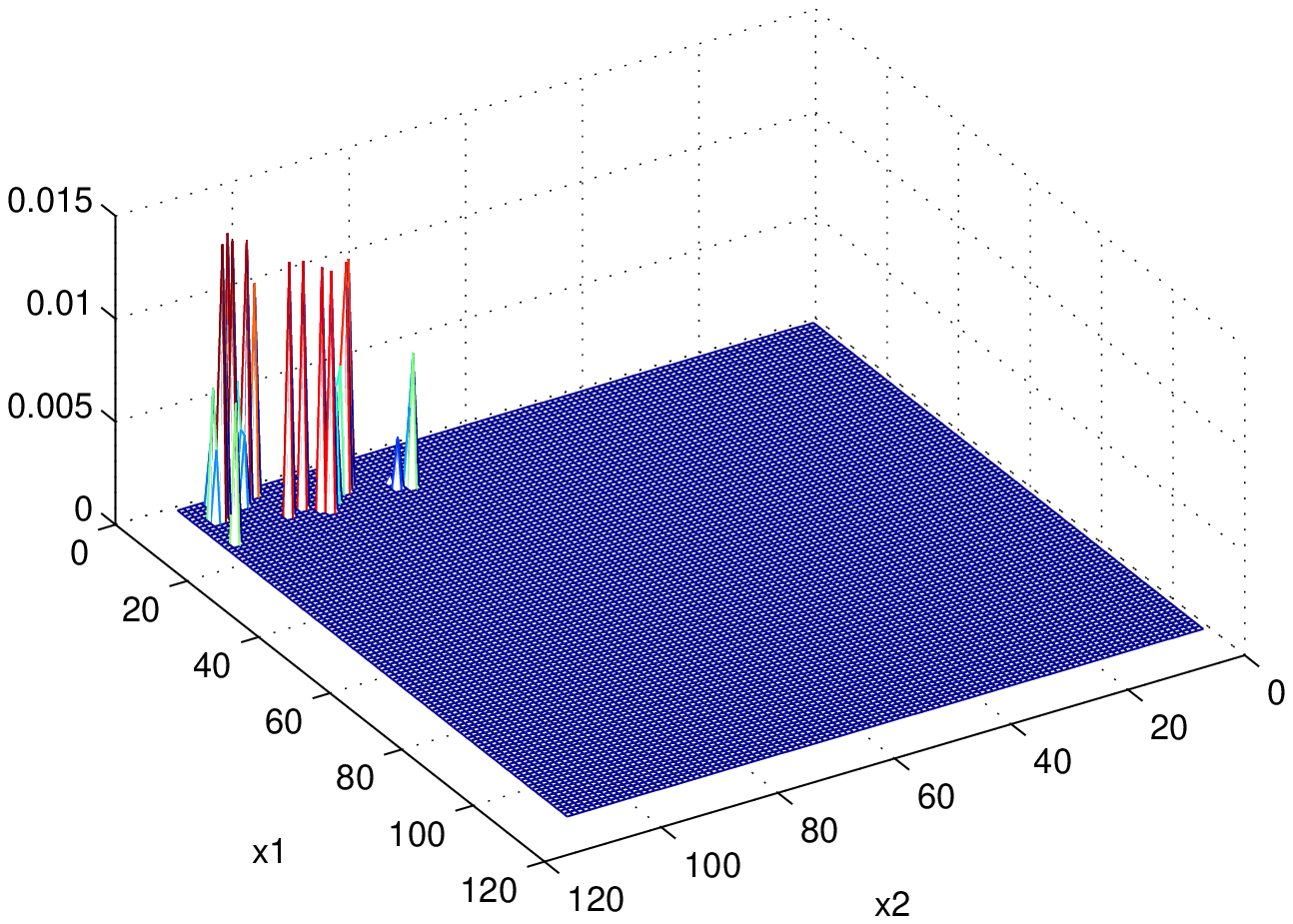}
\caption{PDF of the strategies with switching actions, i.e., function $f_\mathbf{s}^*(x)$ when $\mathbf{s}=\{2,1,0\}$. $x_1$ is the switching time and $x_2$ is the buying time.}
\end{minipage}
\quad
\begin{minipage}[b]{0.32\linewidth}
\centering
\includegraphics[width=1\textwidth]{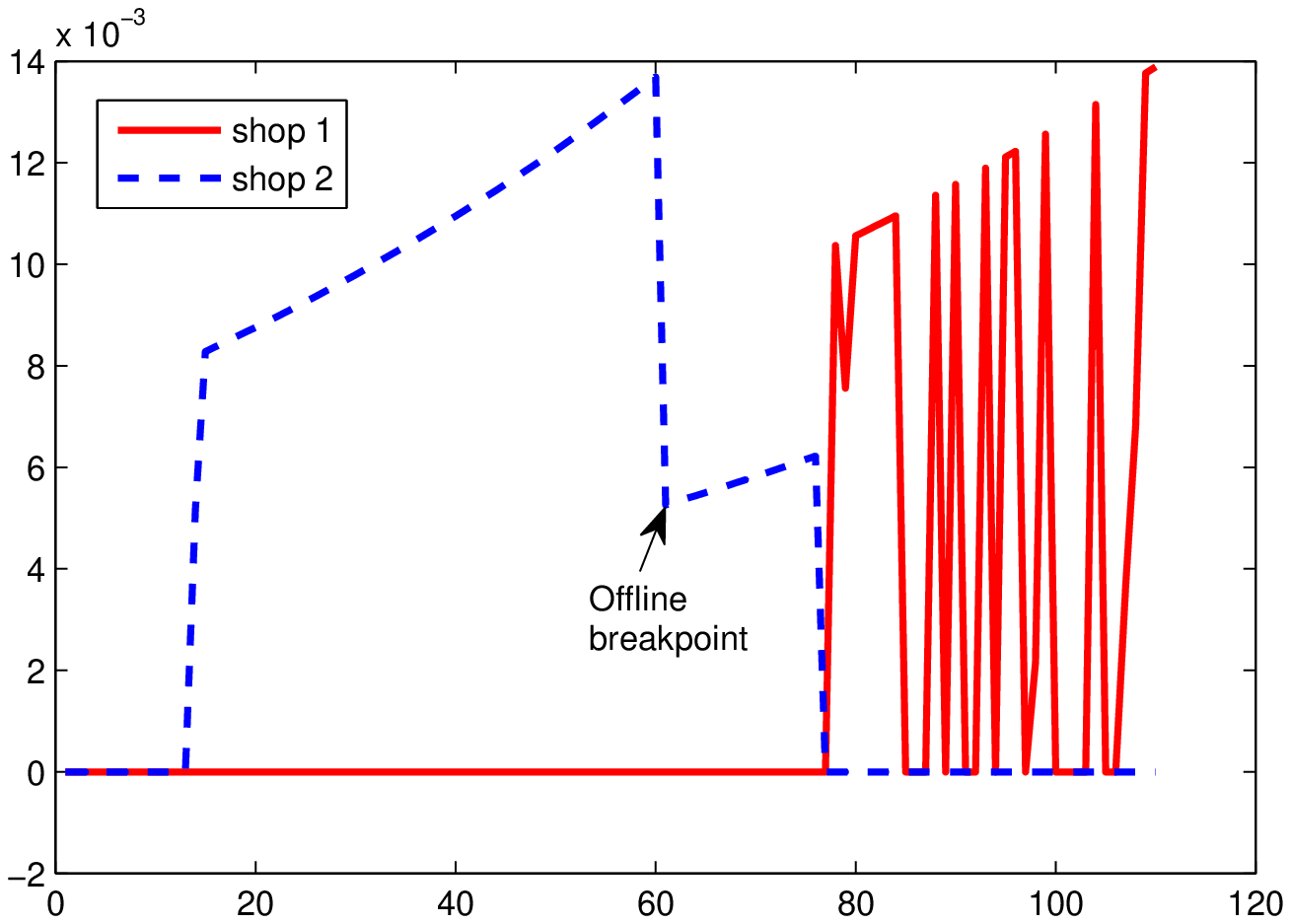}
\caption{PDF of the strategies without switching actions, i.e., function $f_\mathbf{s}^*(x)$ when $|s|=2$. Blue: $f_\mathbf{s}^*(x)$ when $\mathbf{s}=\{2,0\}$; Red: $f_\mathbf{s}^*(x)$ when $\mathbf{s}=\{1,0\}$.}
\end{minipage}
\quad
\begin{minipage}[b]{0.32\linewidth}
\centering
\includegraphics[width=1\textwidth]{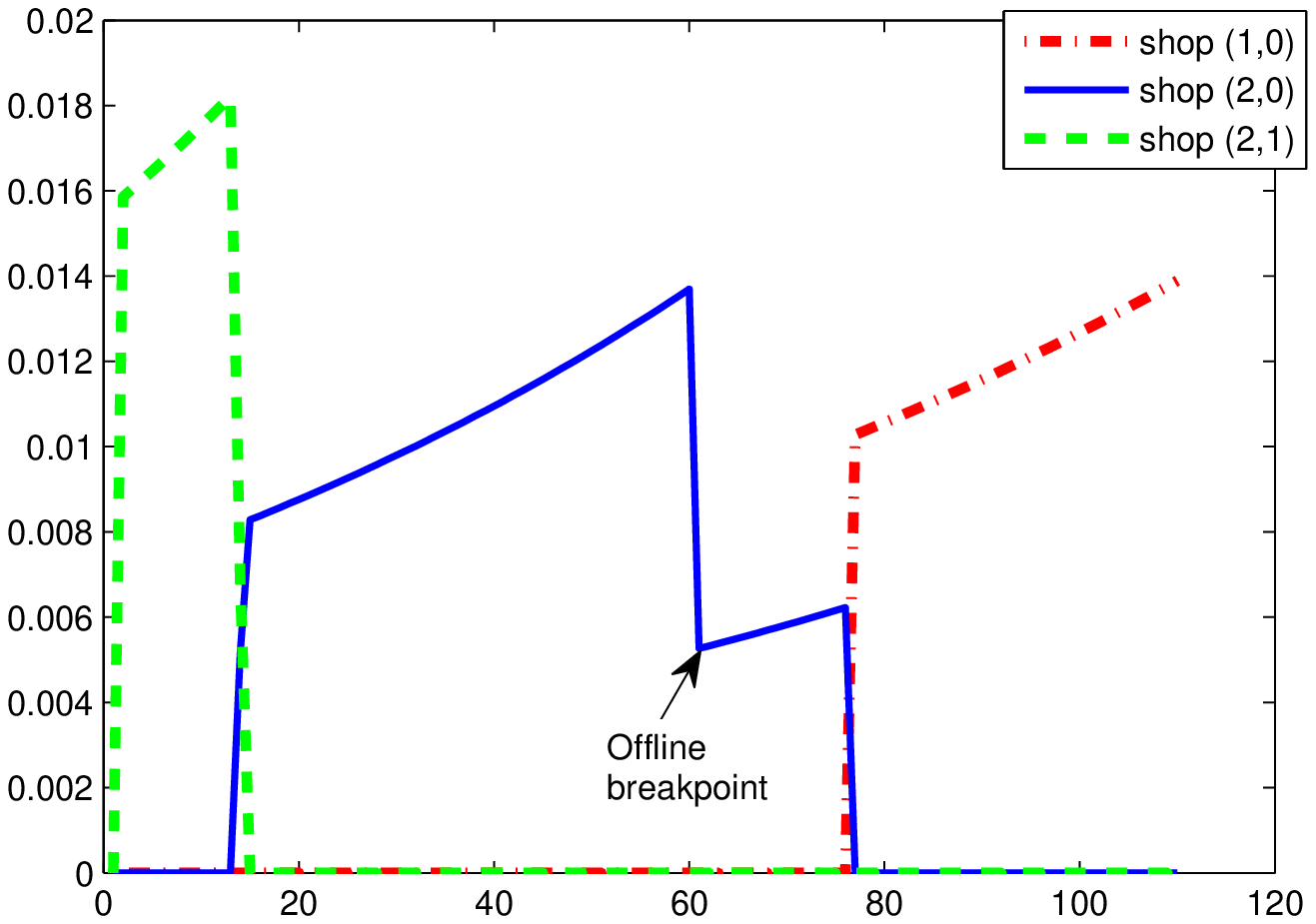}
\caption{PDF of the virtual shops, i.e., function $p_{(i,j)}^{(\mathbf{f})}(x)$ when $(i,j)$ is a virtual shop. Green: ${p_{(2,1)}^*}(x)$; Blue: ${p_{(2,0)}^*}(x)$; Red: ${p_{(1,0)}^*}(x)$.}
\end{minipage}
}
\end{figure*}

Therefore, the target of the problem converts from the optimal $\mathbf{f}$ to the optimal $\mathbf{p}$. Note that for the switching operation from shop $i$ to shop $j$, we do not care about which action $\psi$ it belongs to. Instead, the only thing that matters is when this switching operation happens. This is similar to the one shop case, in which we only care about when the consumer decides to buy. Thus, for each switching pair $(i,j)$, we consider it as a virtual shop. Among these $O(n^2)$ virtual shops, no switching will appear. Thus, we show that the \texttt{MSR-ES} problem is almost the same as \texttt{MSR-E}.

Now we show our settings for virtual shops. For all $i,j\in[n]$ such that $a_i<a_j,r_i>r_j$, we define the virtual shop $(i,j)$ with entry fee $a_{(i,j)}=a_i-a_j$, renting price $r_{(i,j)}=r_i-r_j$ and buying price $b_{(i,j)}=a_j$. We regard the switching time from $i$ to $j$ as the buying time in  virtual shop $(i,j)$. For special case, the prices of virtual shop $(j,0)$ is the same as the real shop $j$. Through this setting, it is not hard to verify that for any action $\psi$ and any $0\leq y\leq B$, the cost function $c(\psi,y)$ is exactly the summation of the cost in the corresponding virtual shops. Similar to a real shop, we define the cost for each virtual shop $(i,j)$:
\begin{eqnarray*}
C_{(i,j)}(\mathbf{p^{(f)}},y)&\triangleq&\int_{0}^{y}(a_{(i,j)}+r_{(i,j)}x+b_{(i,j)})p_{(i,j)}^{(\mathbf{f})}(x)\ud x\\&&+\int_{y}^{B}(a_{(i,j)}+r_{(i,j)}y)p_{(i,j)}^{(\mathbf{f})}(x)\ud x
\end{eqnarray*}
Now we are ready to formalize \texttt{MSR-ES} by the following theorem:
\begin{theorem}
\label{theorem:MSR-ES}
The optimization problem for the consumer can be formalized as follows:
\begin{align}
\mathrm{minimize}&~~ \lambda \label{eqn:f2q}\\
\mathrm{subject ~to}&~~ \frac{C(\mathbf{f},y)}{\mathrm{OPT}(y)}=\lambda \tag{\ref{eqn:f2q}a}\\
&~~C(\mathbf{f},y)={\sum_{(i,j)\in[n]^2:a_i<a_j,r_i>r_j}C_{(i,j)}(\mathbf{p^{(f)}},y)}\nonumber\\
&\quad \quad \quad ~~+{\sum_{j\in[n]}C_{(j,0)}(\mathbf{p^{(f)}},y)},  \quad \forall y\in [0,B] \tag{\ref{eqn:f2q}b}\\
&~~{\sum\limits_{j\in[n]}\int_{0}^{B}p_{(j,0)}^\mathbf{(f)}(x)\ud x=1} \tag{\ref{eqn:f2q}c}\\
&\sum\limits_{j\in[n]:a_j<a_i,r_j>r_i}\int_{y}^{B}p_{(j,i)}^\mathbf{(f)}(x)\ud x \le \int_{y}^{B}p_{(i,0)}^\mathbf{(f)}(x)\ud x \nonumber\\
&+\sum\limits_{j\in[n]:a_i>a_j,r_i<r_j}\int_{y}^{B}p_{(i,j)}^\mathbf{(f)}(x)\ud x,~~\forall i\in[n],y\in(0,B] \tag{\ref{eqn:f2q}d}
\end{align}
\end{theorem}

Thus, \texttt{MSR-ES} can be regarded as \texttt{MSR-E} with $O(n^2)$ shops. The difference is that the summation of the buying probabilities in each virtual shop may be larger than 1. Fortunately, those nice properties of \texttt{MSR-E} still hold for \texttt{MSR-ES}.
\begin{lemma}\label{lemma:MSR-ES seg}
Lemma~\ref{lemma:MSR-E finite} and \ref{lemma:moveP MSR-E} still hold for the virtual shops in the \texttt{MSR-ES} problem.
\end{lemma}

As in other extensions before, the probability density function of the virtual shops is segmented and each segment is an exponential function. The consumer only assigns positive buying probability in exactly one virtual shop at any time. As the buying time increases, she follows the virtual shop order in which the ratio between buying price and renting price is increasing.

In the following 3 figures, we give a simple example when $n=2$ in order to make Lemma~\ref{lemma:MSR-ES seg} easier to understand. We approach the optimal strategy through the discrete model and the figures show the p.d.f. functions in the optimal strategy. It can be seen that Lemma~\ref{lemma:MSR-ES seg} is verified. Since there is a proof for Lemma~\ref{lemma:MSR-ES seg}, we do not give more complicated examples. The parameters for the 2 shops are as follows: $a_1=80, r_1=1, b_1=110, a_2=20, r_2=2, b_2=180$.

\section{Conclusions}
In this paper, we consider the multi-shop ski rental problem (\texttt{MSR}) and its extensions (\texttt{MSR-S}, \texttt{MSR-E}, and \texttt{MSR-ES}), in which there are multiple shops and the consumer wants to minimize the competitive ratio.

For each problem, we prove that in the optimal mixed strategy of the consumer, she only assigns positive buying probability to \emph{exactly one} shop at any time. The shop order strongly relates to the ratio between buying price and renting price, even in which entry fee is involved. Further, in the basic problem (\texttt{MSR}), we derive a linear time algorithm for computing the optimal strategy of the consumer. For \texttt{MSR-S}, we prove that under the optimal mixed strategy, the consumer only switches to another shop at the buying time.

In problems \texttt{MSR-E} and \texttt{MSR-ES}, we show that the optimal strategy can be solved if the breakpoints are known. Similar to the basic problem (\texttt{MSR}), we conjecture that the quasi-concave property also holds for these two variants. Further, we conjecture that there exists an iteration algorithm using gradient decent technique, which might converge to the optimal solution.



%
\bibliographystyle{abbrv}
\bibliography{sigproc}  
%
%
\appendix
\section{PROOF OF THEOREM 1}

\subsection{Proof of Lemma~\ref{lemma:MSRconstant}}
Denote the optimal game value achieved by the optimal strategy $\mathbf{p^*}$ by $\lambda_{\mathbf p^*}$. For any strategy $\mathbf{p}\in \calP$ and any $y \in (0,B]$, denote $\lambda_{\mathbf{p}}(y)$ to be the difference between the maximum competitive ratio between the ratio on $y$, i.e.,
\begin{displaymath}
\lambda_{\mathbf p}(y) \triangleq \sup_{y\in (0,B]} \left\{\frac{C(\mathbf{p},y)}{r_1 y}\right\} - \frac{C(\mathbf{p},y)}{r_1 y}
\end{displaymath}
By definition, we know that $\lambda_{\mathbf{p}}(y)\geq 0, \forall \mathbf{p}\in \calP, y \in (0,B] $.

To prove this lemma, we first to state the following propositions:
\begin{proposition}
\label{proposition:con1}
For any $\mathbf{p}\in \calP$ and $y_0 \in (0,B]$, if there exists a constant $\delta>0$, such that $\forall y\in (y_0-\delta, y_0)$, $\lambda_{\mathbf{p}}(y)>\lambda_{\mathbf{p}}(y_0)$, then at least one of the following properties is satisfied:
\begin{enumerate}
\item There exists some $j\in [n]$ such that $p_j(y_0)$ is a probability mass.
\item There exists some $j \in [n]$, such that $\int_{y_0-\delta}^{y_0} p_j(t)\ud t > 0$.
\end{enumerate}
\end{proposition}

\begin{proof}
If both properties are unsatisfied, then $\forall j \in [n]$, $p_{j,y_0}=0$ and  $ \int_{y_0-\delta}^{y_0}p_j(t)\ud t=0$. Then for any $j\in[n]$ and any $y\in (y_0-\delta, y_0)$, we have
\begin{eqnarray*}
C_j(p_j, y) &=& \int_0^{y_0-\delta} (r_j t+b_j)p_j(t) \ud t+ \int_{y_0}^{B} r_j yp^*_j(t) \ud t\\
C_j(p_j, y_0) &=& \int_0^{y_0-\delta} (r_j t+b_j)p_j(t) \ud t+ \int_{y_0}^{B} r_j y_0 p^*_j(t) \ud t
\end{eqnarray*}
This directly derives that
\begin{displaymath}
\frac{C(\mathbf p, y)}{r_1 y} \geq \frac{C(\mathbf p, y_0)}{r_1 y_0}
\end{displaymath}
which contradicts the fact that $\lambda_{\mathbf{p}}(y) > \lambda_{\mathbf{p}}(y_0)$.
\end{proof}

\begin{proposition}
\label{proposition:con2}
For any $\mathbf{p}\in \calP$ and $0<a<b$, if $\forall y\in (a,b)$, $\lambda_{\mathbf{p}}(y)>\lambda_{\mathbf{p}}(a)$, then there exists some $j\in [n]$, such that $\int_0^{a}p_j(t)\ud t>0$.
\end{proposition}

\begin{proof}
Assume by contradiction that $\forall j\in [n]$, we have $\int_0^{a}p_j(t)\ud t=0$. For any $y_0\in(a,a+\min\{\frac{b-a}{2},\frac{b_n}{2r_n}\})$, we have
\begin{align*}
&\frac{C(\mathbf p, y_0)}{r_1 y_0} - \frac{C(\mathbf p, a)}{r_1 a}\\
\l\ge&\frac{\sum_j \int_{a}^{y_0} (r_j t+b_j)p_j(t) \ud t}{r_1 y_0}-\frac{\sum_j \int_{a}^{y_0} r_j p_j(t) \ud t}{r_1 }
\geq 0
\end{align*}
The last inequality is because $r_j t+b_j>r_j y_0$ for $a< t\leq y_0$. But note that $\lambda_{\mathbf{p}}(y_0)>\lambda_{\mathbf{p}}(a)$ which makes a contradiction. Thus, we prove the proposition.
\end{proof}

\begin{proposition}
\label{proposition:con3}
For any $\mathbf{p}\in \calP$, any $\epsilon>0$ and $y_0\in (0,B]$, there exists a constant $\delta\in (0,y_0/4)$, such that for any $y \in (y_0-\delta, y_0)$, $\lambda_{\mathbf p}(y) \geq \lambda_{\mathbf p}(y_0)-\epsilon$.
\end{proposition}
\begin{proof}
Let $\delta=\frac{\epsilon r_1 y_0^2}{4(r_ny_0+b_1+\epsilon r_1 y_0)}$. We first prove that $\Delta=\frac{c_j(x,y)}{r_1 y}-\frac{c_j(x,y_0)}{r_1 y_0}\leq \epsilon$, for any action $(j,x)\in \Psi_c$ and any $y\in (y_0-\delta, y_0)$. Note that $\Delta$ can be negative.
\begin{enumerate}
\item If $x\leq y$, we have
$$
\Delta=\frac{r_j x+b_j}{r_1 y}-\frac{r_j x+b_j}{r_1 y_0}\leq \frac{\delta(r_j x+b_j)}{r_1(y_0-\delta)y_0}\leq\frac{\epsilon r_1 y_0^2/4}{3r_1y_0^2/4}\leq \epsilon
$$
\item If $y<x\leq y_0$, we have
$$
\Delta=\frac{r_j y}{r_1 y}-\frac{r_j x+b_i}{r_1 y_0}=\frac{r_j(y_0-x)-b_j}{r_1 y_0}\leq \frac{r_j\delta-b_j}{r_1 y_0}\leq \epsilon
$$
\item If $y_0<x$, we have
$$
\Delta=\frac{r_j y}{r_1 y}-\frac{r_j y_0}{r_1 y_0}=0\leq\epsilon
$$
\end{enumerate}

Thus, we prove $\frac{c_j(x,y)}{r_1 y}-\frac{c_j(x,y_0)}{r_1 y_0}\leq \epsilon$, $\forall (j,x)\in\Psi_c$ and $\forall y\in(y_0-\delta,y_0)$. This implies that $\forall \mathbf{p}\in\mathcal P$ and $\forall y\in(y_0-\delta,y_0)$, we have
$$\frac{C(\mathbf{p},y)}{r_1 y}-\frac{C(\mathbf{p},y_0)}{r_1 y_0}
=\sum_{j\in[n]}\int_{0}^{B}(\frac{r_j y}{r_1 y}-\frac{r_j y_0}{r_1 y_0})p_j(t)\ud t\leq \epsilon$$
From the above inequality, we obtain $\lambda_{\mathbf p}(y)\ge\lambda_{\mathbf p}(y_0)-\epsilon$.
\end{proof}

Now we begin to prove this lemma.

\begin{proof}
Assume by contradiction that there exists some $y_0 \in (0,B]$ such that $\lambda_{\mathbf p^*}(y_0) = 2\Lambda >0$. By Proposition~\ref{proposition:con3}, there exists a constant $0<\delta <y_0/4$, such that for any $y \in (y_0-\delta, y_0)$, $\lambda_{\mathbf p^*}(y) \geq \Lambda$.

Define two sets $L$ and $R$ as follows:
\begin{eqnarray*}
L &\triangleq& \{y_1: y_1 \leq y_0 - \delta, \lambda_{\mathbf p^*}(y_1)\leq \Lambda/8 \}\\
R &\triangleq& \{y_2: y_2 \geq y_0, \lambda_{\mathbf p^*}(y_2)\leq \Lambda/8 \}
\end{eqnarray*}
Since function $\lambda_{\mathbf p^*}$ is right continuous and Proposition~\ref{proposition:con3}, observe that $L$ and $R$ are union sets of closed intervals except that $L$ might contain a half-closed interval $(0,y]$. We further denote $y_L$ and $y_R$ as follows:
\begin{eqnarray*}
y_L &\triangleq& \max L\\
y_R &\triangleq& \min R
\end{eqnarray*}
Since $\forall y\in(y_0-\delta,y_0),\lambda_{\mathbf{p^*}}(y)\ge\Lambda$, it is impossible that both $L$ and $R$ are empty. According to whether $L$ and $R$ are empty, we can break the proof into three cases. In each case, we construct another strategy $\mathbf{p}^1$ with lower worst competitive ratio, which contradicts the optimality of $\mathbf{p^*}$.

\begin{enumerate}
\item[(a)] $L \neq \emptyset$ while $R = \emptyset$: By Proposition~\ref{proposition:con2}, note that there exists some $j \in [n]$, such that $\int_0^{y_L}p_j^*(t)\ud t>0$. Define $\lambda_0 \triangleq   \lambda_{\mathbf{p^*}}(B)$, and $\epsilon_L \triangleq \min\{\sum_{j\in [n]}\int_0^{y_L}p_j^*(t)\ud t,\frac{\lambda_0 r_1}{2b_1}\}$. By definition, note that $\lambda_0\geq \Lambda/8$, and $0\leq y_1\leq y_L$.

    Further, we define $y_1=\arg \min_{ y} \{\sum_{j\in [n]}\int_0^{y}p_j^*(t)\ud t\geq \epsilon_L\}$. We assume $k$ is the minimum shop index such that
    \begin{itemize}
    \item $\sum_{j=1}^{k-1}\int_0^{y_1}p_j^*(t)\ud t+\sum_{j=k}^n\int_0^{y_1^-}p_j^*(t)\ud t\leq \epsilon_L$,
    \item $\sum_{j=1}^{k}\int_0^{y_1}p_j^*(t)\ud t+\sum_{j=k+1}^n\int_0^{y_1^-}p_j^*(t)\ud t\geq \epsilon_L$.
    \end{itemize}

We construct $\mathbf{p}^1$ as follows:
\begin{eqnarray*}
p_{j,y_1}^1 &=& 0, \quad \forall 1\leq j\leq k-1 \\
p_{k,y_1}^1 &=& p_{k,y_1}^* - \epsilon_L+\sum_{j=1}^{k-1}\int_0^{y_1}p_j^*(t)\ud t+\sum_{j=k}^n\int_0^{y_1^-}p_j^*(t)\ud t \\
p_j^1(y) &=& 0, \quad \forall j\in [n], ~y \in [0, y_1)\\
p_{1,B}^1 &=&  p_{1,B}^*+\epsilon_L\\
p_{j}^1(y) &=& p_{j}^*(y), \quad \textrm{elsewhere}
\end{eqnarray*}

Now we prove that $\lambda_{\mathbf{p}^1} < \lambda_{\mathbf{p^*}}$.

If $0< y< B$, it holds that
\begin{eqnarray*}
\label{inequality:maxv}
&&\frac{C(\mathbf{p}^1,y)}{r_1 y} - \frac{C(\mathbf{p}^*,y)}{r_1 y}\\
&\leq& \frac{ \sum_{j\in [n]}\int_{0}^{y} (r_1 y - r_j t -b_j) (p_j^*(t)-p_j^1(t)) \ud t}{r_1 y} \leq 0
\end{eqnarray*}
where the last inequality holds since that for any $0\leq t\leq y$ and $j\in [n]$, we have $r_1 y - r_j t -b_j < b_n - b_j \leq 0$, and $ p_j^*(t)-p_j^1(t) \geq 0$.

Assume $y_2=\inf \{\arg \min_{y \in L} \lambda_{\mathbf p^*}(y)\}$. If $y_2>0$, by proposition~\ref{proposition:con2}, we have $\sum_{j\in [n]}\int_{0}^{y} (p_j^*(t)-p_j^1(t))\ud t > 0$ for $y\geq y_2$. Thus, inequality~\ref{inequality:maxv} is strictly less while $y_2\leq y<B$.

Otherwise, $y_2=0$. By the definition of $y_2$ and function $\lambda_{\mathbf p^*}$ is right continuous, there must be some $0<y_3< y_1$, such that for any $0<y\leq y_3$, $\lambda_{\mathbf p^*}(y)$ is a constant. Then by Proposition~\ref{proposition:con2}, we have $\sum_{j\in [n]}\int_0^{y_3}(p_j^*(t)-p_j^1(t))\ud t > 0$. Thus, we have $\frac{C(\mathbf{p}^1,y_3)}{r_1 y_3} - \frac{C(\mathbf{p}^*,y_3)}{r_1 y_3}<0$. Note that $p_j^1(y)=0$ for any $j\in [n]$ and $y\in [0,y_3]$. It is obvious that $\frac{\lambda_{\mathbf{p}^1}(y)}{r_1 y}=\frac{\lambda_{\mathbf{p}^1}(y_3)}{r_1 y_3}$ holds for any $y\in (0,y_3]$. Thus, inequality~\ref{inequality:maxv} is strictly less while $0<y<B$.

If $y=B$, it holds that
\begin{eqnarray*}
&&\frac{C(\mathbf{p}^1,B)}{r_1 B} - \frac{C(\mathbf{p}^*,B)}{r_1 B}\\
&=& \frac{ \sum_{j\in [n]}\int_{0}^{y_L} (r_1 B+b_1 - r_j t -b_j) (p_j^*(t)-p_j^1(t)) \ud t}{r_1 B} \\
&\leq & \frac{b_1 B\epsilon_L}{r_1 B}
\leq \frac{\lambda_0}{2}
\end{eqnarray*}
Thus we have $\lambda_{\mathbf{p}^1} < \lambda_{\mathbf{p^*}}$.

\item[(b)]  $L = \emptyset$ while $R \neq \emptyset$: Define $\lambda_0 \triangleq \lambda_{\mathbf{p^*}}(y_0-\frac{\delta}{2})$, $\epsilon_R \triangleq \frac{\lambda_0 r_1}{2 b_1}$. Note that there exists some $j \in [n]$, such that $\int_{y_0}^{y_R}p_j^*(y)\ud y>0$ by Proposition~\ref{proposition:con1}. We construct $\mathbf{p}^1$ as follows:

If $p_{j,y_R}^*$ a probability mass. We construct $\mathbf{p}^1$ as follows:
\begin{eqnarray*}
p_{j,y_R}^1 &=& p_{j,y_R}^* - \min \{p_{j,y_R}^*, \epsilon_R\} \\
p_{1,y_0}^1 &=&\min \{p_{j,y_R}^*, \epsilon_R\}+ p_{1,y_0}^* \\
p_{1}^1(y) &=& p_{1}^*(y), \quad \textrm{elsewhere}
\end{eqnarray*}

Otherwise, we construct $\mathbf{p}^1$ as follows:
\begin{eqnarray*}
p_j^1(y) &=& p_j^*(y)- \min\{p_j^*(y), \frac{\epsilon_R}{y_R-y_0}\}, \quad y \in (y_0, y_R]\\
p_{j,y_0-\frac{\delta}{2}}^1 &=& p_{j,y_0-\frac{\delta}{2}}^*+\int_{y_0}^{y_R}\min\{p_j^*(t), \frac{\epsilon_R}{y_R-y_0}\}\ud t \\
p_{j}^1(y) &=& p_{j}^*(y),\quad \textrm{elsewhere}
\end{eqnarray*}
Now we prove that $\lambda_{\mathbf{p}^1} < \lambda_{\mathbf{p^*}}$.

If $0<y< y_0-\frac{\delta}{2}$, we have
$$
\frac{C(\mathbf{p}^1,y)}{r_1 y} - \frac{C(\mathbf{p}^*,y)}{r_1 y}
= \frac{\int_{(y_0-\frac{\delta}{2})^-}^{B}  r_j y (p^1_j(t)-p^*_j(t)) \ud t}{r_1 y} = 0.
$$

If $y>y_0-\frac{\delta}{2}$, we have
$$
\frac{C(\mathbf{p}^1,y)}{r_1 y} - \frac{C(\mathbf{p}^*,y)}{r_1 y}
= \frac{\int_{y}^{B}  r_j y (p^1_j(t)-p^*_j(t)) \ud t}{r_1 y}
\leq 0.
$$
Specially, the inequality is strictly less if $y\geq y_R$.

If $y=y_0-\frac{\delta}{2}$ and $p_{j,y_R}^*$ is a probability mass, we have

$$
\frac{C(\mathbf{p}^1,y)}{r_1 y} - \frac{C(\mathbf{p}^*,y)}{r_1 y}
=  \frac{b_j \min \{p_{j,y_R}^*, \epsilon_R\}}{r_1 y}
\leq\frac{\lambda_0}{2}.
$$

Otherwise, $y=y_0-\frac{\delta}{2}$ and $p_{j,y_R}^*$ is not a probability mass, we have
$$
\frac{C(\mathbf{p}^1,y)}{r_1 y} - \frac{C(\mathbf{p}^*,y)}{r_1 y}
=  \frac{b_j\int_{y_0}^{y_R} \min\{p_j^*(t), \frac{\epsilon_R}{y_R-y_0}\} \ud t}{r_1 y}
\leq\frac{\lambda_0}{2}.
$$

Thus we have $\lambda_{\mathbf{p}^1} < \lambda_{\mathbf{p^*}}$.

\item[(c)] $L = \emptyset$ while $R \neq \emptyset$: 
    Define $\lambda_0 \triangleq \lambda_{\mathbf{p^*}}(y_0-\frac{\delta}{2})$. In this case, we construct $\mathbf{p}^1$ by combining the settings in the two cases above.
    
    Firstly, we still let $\epsilon_R \triangleq \frac{\lambda_0 r_1}{2b_1}$, and construct a new strategy $\mathbf{p'}$ using the same setting of case (b). By the argument above, we have that $\frac{C(\mathbf{p'},y)}{r_1 y}$ stays while $y\in (0,y_0-\frac{\delta}{2}]$, and decreases while $y\in (y_0-\frac{\delta}{2},B]$. 
    Then we suppose $\frac{C(\mathbf{p'},B)}{r_1 B} - \frac{C(\mathbf{p}^*,B)}{r_1 B}=-\epsilon '$, and set $\epsilon_L$ as follows:
$$
\epsilon_L=\min\{\sum_{j\in [n]}\int_0^{y_L}p_j^*(t)\ud t,\frac{r_1\epsilon'}{2b_1}\}. 
$$

Use the same setting of case (a), we construct a new strategy $\mathbf{p}^1$ from $\mathbf{p'}$. By the argument in the case (a), we observe that $\frac{C(\mathbf{p}^1,B)}{r_1 B} - \frac{C(\mathbf{p'},B)}{r_1 B}\leq \frac{\epsilon '}{2}$. Thus, we have that $\frac{C(\mathbf{p}^1,B)}{r_1 B} - \frac{C(\mathbf{p}^*,B)}{r_1 B}\leq -\frac{\epsilon '}{2}$. Further, $\frac{C(\mathbf{p}^1,y)}{r_1 y}$ decreases almost everywhere except $y=y_0-\frac{\delta}{2}$, and $\frac{C(\mathbf{p}^1,y_0-\frac{\delta}{2})}{r_1 (y_0-\frac{\delta}{2})}$ increases at most $\frac{\lambda_0}{2}$. By the same argument above, we have $\lambda_{\mathbf{p}^1} < \lambda_{\mathbf{p^*}}$.
\end{enumerate}
Thus, in all the cases above, we have proved that $\lambda_{\mathbf{p}^1} < \lambda_{\mathbf{p}^*}$, which contradicts the fact that $\mathbf{p}^*$ is optimal. 
\end{proof}

\subsection{Proof of Lemma~\ref{lemma:MSRfinite}}
\begin{proof}
Of (\ref{ratioRelation:MSR2}), multiply both sides by $r_1 y$ and then take derivatives, we will get
\begin{equation} \label{PfirstDiff}
\sum_{j=1}^{n} \left(b_j p^*_j(y)+\int_{y}^{B}r_jp^*_j(x)\ud x\right)=\lambda r_1
\end{equation}
From this equation, we can directly get that $p^*_j(y) \leq \frac{\lambda r_1}{b_n}, \forall y\in(0,B]$. Note that the pure strategy $(1,B)\in\Psi_c$ has a competitive ratio of $\frac{b_1+b_n}{b_n}$ and, according to Lemma~\ref{lemma:MSRconstant}, the optimal randomized strategy must be better. So the optimal strategy has a competitive ratio $\lambda< \frac{2b_1}{b_n}$. Thus, we have proved that $\forall x\in(0,B],p^*_j(x)<\frac{2b_1 r_1}{b_n^2}$.

Assume that there exists some $j$ such that $p^*_j(0) = + \infty$, then $\lim_{y\rightarrow 0^+} \frac{C(\mathbf{p^*}, y)}{r_1 y} = +\infty$ which contradicts (\ref{ratioRelation:MSR2}). Thus $p^*_j(0) < +\infty$ for any $j=1,\cdots,n$.
\end{proof}

\subsection{Proof of Lemma~\ref{lemma:moveP MSR}}
For an arbitrary $x$, if there exists an $r>0$, such that $\int_{x-r}^{x+r} p_j(t)\ud t=0$, we say $p_j(x)=0$. Since $p_j(x)$ is a probability density function, it does not lose the generality. Hence, to prove this lemma, it suffices to show that $\forall x\in (0,B)$, $\forall \epsilon\in(0,x)$, $\forall j\in[n]$,  if $\int_{x-\epsilon}^{x}p_j^*(x) \ud x> 0$, then $\forall j'>j, x'\ge x$, we must have $\int_{x'}^{B}  p_{j'}^*(x')=0$. To prove it, we assume by contradiction that if $\int_{x-\epsilon}^{x}p_j^*(t)\ud t>0$ and $\int_{x'}^{x'+\epsilon}p_{j'}^*(t)\ud t>0$. And first we need the following proposition:

\begin{proposition}\label{propostion:premoveP}
If $\int_{x-\epsilon}^{x}p_j^*(t)\ud t>0$ and $\int_{x'}^{x'+\epsilon}p_{j'}^*(t)\ud t>0$, then there exist $x_1, x_2$ and $\epsilon_0>0$, such that $(x_1, x_1+\epsilon_0)\subset (x-\epsilon,x) $, $ (x_2, x_2+\epsilon_0)\subset (x', x'+\epsilon)$ and $\int_{0}^{\epsilon_0} \min\{p_j^*(x_1+\theta),p_{j'}^*(x_2+\theta)\} d\theta>0$.
\end{proposition}
\begin{proof}
We denote by $\delta_1>0$ the value of $\int_{x-\epsilon}^{x}p_j^*(t)\ud t/\epsilon$ and by $\delta_2>0$ the value of $\int_{x'}^{x'+\epsilon}p_{j'}^*(t)\ud t/\epsilon$.
Let $T=2b_1r_1/b_n^2$ and, by Lemma~\ref{lemma:MSRfinite},
$\forall j\in[n],\forall t\in(0,B]$, we have $p^*_j(t)<T$. Given 2 independent random variables $t_1\sim U(x-\epsilon,x)$ and $t_2\sim U(x',x'+\epsilon)$, we have $\delta_1=\mathbb E[p_j^*(t_1)]$ and $\delta_2=\mathbb E[p_{j'}^*(t_2)]$. Therefore,
\begin{align*}
&\operatorname{Pr}[p_{j}(t_1)>\delta_1/2]\cdot T+(1-\operatorname{Pr}[p_{j}(t_1)>\delta_1/2])\cdot\delta_1/2\geq \delta_1 \\
&\operatorname{Pr}[p_{j'}(t_2)>\delta_2/2]\cdot T+(1-\operatorname{Pr}[p_{j'}(t_2)>\delta_1/2])\cdot\delta_2/2\geq \delta_2
\end{align*}
Let $\delta=\min\{\delta_1/2, \delta_2/2\}>0$ and from the above 2 inequalities, we obtain
\begin{align*}
&\operatorname{Pr}[p_{j}(t_1)>\delta]\ge \delta/T\\ &\operatorname{Pr}[p_{j'}(t_2)>\delta]\ge \delta/T
\end{align*}

Now we focus on $\int_{x'}^{x'+\epsilon}\ud t' \int_{x-\epsilon}^{x}\min\{p_j^*(t),p_{j'}^*(t')\}\ud t$:
\begin{align}
&\int_{x'}^{x'+\epsilon}\ud t' \int_{x-\epsilon}^{x}\min\{p_j^*(t),p_{j'}^*(t')\}\ud t \nonumber\\
\geq & \epsilon \operatorname{Pr}[p_{j'}(t_2)>\delta] \int_{x-\epsilon}^{x}\min\{p_j^*(t'),\delta\}\ud t \nonumber \\
\geq & \epsilon \delta (\epsilon \operatorname{Pr}[p_{j}(t_1)>\delta]\delta)/T \nonumber \\
=& \epsilon^2 \delta^3/T^2 >0 \label{ineq:prob1}
\end{align}
On the other hand, we have
\begin{align}
&\int_{x'}^{x'+\epsilon}\ud t' \int_{x-\epsilon}^{x}\min\{p_j^*(t),p_{j'}^*(t')\}\ud t \nonumber \\
=&\int_{x-\epsilon}^{x}\ud t \int_{x'}^{x'+\epsilon}\min\{p_j^*(t),p_{j'}^*(t')\}\ud t' \nonumber \\
=&\int_{x-\epsilon}^{x}\ud t \int_{x'}^{x'+x-t}\min\{p_j^*(t),p_{j'}^*(t'+\epsilon+t-x)\}\ud t' \nonumber \\
&+\int_{x-\epsilon}^{x}\ud t \int_{x'+x-t}^{x'+\epsilon}\min\{p_j^*(t),p_{j'}^*(t'+t-x)\}\ud t' \nonumber \\
=&\int_{x'}^{x'+\epsilon}\ud t' \int_{x-\epsilon}^{x-(t'-x')}\min\{p_j^*(t),p_{j'}^*(t'+\epsilon+t-x)\}\ud t \nonumber \\
&+\int_{x'}^{x'+\epsilon}\ud t' \int_{x-(t'-x')}^{x}\min\{p_j^*(t),p_{j'}^*(t'+t-x)\}\ud t \label{ineq:prob2}
\end{align}
Combine (\ref{ineq:prob1}) and (\ref{ineq:prob2}), it is proved that there exists some $t'\in (x', x'+\epsilon)$, such that either $\int_{x-\epsilon}^{x-(t'-x')}\min\{p_j^*(t),p_{j'}^*(t'+\epsilon+t-x)\}\ud t>0$ or $\int_{x-(t'-x')}^{x}\min\{p_j^*(t),p_{j'}^*(t'+t-x)\}\ud t >0$.

\textbf{Case 1}: if $\int_{x-\epsilon}^{x-(t'-x')}\min\{p_j^*(t),p_{j'}^*(t'+\epsilon+t-x)\}\ud t>0$, then we set $x_1=x-\epsilon, x_2=t',\epsilon_0=\epsilon-(t'-x')$. We directly obtain $(x_1, x_1+\epsilon_0)\subset (x-\epsilon,x) $, $ (x_2, x_2+\epsilon_0)\subset (x', x'+\epsilon)$ and $\int_{0}^{\epsilon_0} \min\{p_j^*(x_1+\theta),p_{j'}^*(x_2+\theta)\} d\theta>0$.

\textbf{Case 2}: if $\int_{x-(t'-x')}^{x}\min\{p_j^*(t),p_{j'}^*(t'+t-x)\}\ud t >0$, then we set $x_1=x-(t'-x'), x_2=x',\epsilon_0=t'-x'$. We directly obtain $(x_1, x_1+\epsilon_0)\subset (x-\epsilon,x) $, $ (x_2, x_2+\epsilon_0)\subset (x', x'+\epsilon)$ and $\int_{0}^{\epsilon_0} \min\{p_j^*(x_1+\theta),p_{j'}^*(x_2+\theta)\} d\theta>0$.
\end{proof}

And to prove this lemma, it suffices to prove the following proposition:

\begin{proposition}
$\forall x\in (0,B)$, $\forall \epsilon>0$, $\forall j\in[n]$ and $\forall j'>j, x'\ge x$,  it is impossible that both of the following inequalities hold:
\begin{itemize}
\item $\int_{x-\epsilon}^{x}p_j^*(x) \ud x> 0$
\item $\int_{x'}^{x'+\epsilon}  p_{j'}^*(x')\ud x'>0$
\end{itemize}  
\end{proposition}

\begin{proof}
We assume by contradiction that there exist some $j'>j, x'\ge x,\epsilon>0$ such that both of the above inequalities hold. Then we derive a new strategy $\mathbf p^1$ from $\mathbf p^*$ and we are going to show that the expected cost of the consumer following $p^1$ is not greater than the cost of $p^*$ for any $y$ and is less than the cost of $p^*$ for some $y$, which contradicts the assumption.

$\forall \theta\in(0,\epsilon_0)$, we adopt the following setting:
\begin{eqnarray*}
p_{j}^1(x_1+\theta) &=& p_j^*(x_1+\theta) - \mu'(\theta)\\
p_{j}^1(x_2+\theta) &=& p_{j}^*(x_2+\theta)+ \mu'(\theta)\\
p_{j'}^1(x_1+\theta)&=&p_{j'}^*(x_1+\theta)+ \mu(\theta)\\
p_{j'}^1(x_2+\theta) &=& p_{j'}^*(x_2+\theta) -  \mu(\theta)
\end{eqnarray*}
where $\mu(\theta)= \frac{b_j+r_j}{b_j+r_j+b_{j'}+r_{j'}}
\min\{p_j^*(x_1+\theta),p_{j'}^*(x_2+\theta)\}$, and $\mu'(\theta)=\frac{b_{j'}+r_{j'}}{b_j+r_j+b_{j'}+r_{j'}} \min\{p_j^*(x_1+\theta),p_{j'}^*(x_2+\theta)\}$.
Note that we break $(0,B]$ into five sub-intervals. We show that for each sub-interval,
$$
C(\mathbf p^*,y)-C(\mathbf p^1,y)\geq 0
$$

When $y \leq x_1$:
\begin{eqnarray*}
&&C(\mathbf p^*,y)-C(\mathbf p^1,y) \\
&=& \int_{0}^{\epsilon_0} r_{j} y  \mu'(\theta) \ud \theta - \int_{0}^{\epsilon_0} r_{j} y  \mu'(\theta) \ud \theta\\
&&- \int_{0}^{\epsilon_0} r_{j'} y \mu(\theta) \ud \theta + \int_{0}^{\epsilon_0} r_{j'} y \mu(\theta) \ud \theta
\\
& =& 0
\end{eqnarray*}

When $x_1 < y \leq x_1+\epsilon_0:$
\begin{eqnarray*}
&&C(\mathbf p^*,y)-C(\mathbf p^1,y) \\
&=& \int_{0}^{y-x_1} [r_{j} (x_1+\theta)+b_{j}] \mu'(\theta) \ud \theta + \int_{y-x_1}^{\epsilon_0} r_{j} y \mu'(\theta) \ud \theta \\
&&- \int_{0}^{\epsilon_0} r_{j} y \mu'(\theta) \ud \theta -\int_{0}^{y-x_1} [r_{j'} (x_1+\theta)+b_{j'}] \mu(\theta) \ud \theta \\
&&- \int_{y-x_1}^{\epsilon_0} r_{j'} y \mu(\theta) \ud \theta
+\int_{0}^{\epsilon_0} r_{j'} y \mu(\theta) \ud \theta\\
&=& \mu_1 \frac{r_{j'}b_j-r_jb_{j'}}{b_j+r_j+b_{j'}+r_{j'}} \geq 0
\end{eqnarray*}
where $\mu_1 = \int_{0}^{y-x_1}(1+y-x_1-\theta)\min\{p_j^*(x_1+\theta),p_{j'}^*(x_2+\theta)\}\ud \theta$.

When $x_1+\epsilon_0 < y \leq x_2:$
\begin{eqnarray*}
&&C(\mathbf p^*,y)-C(\mathbf p^1,y) \\
&=&  \int_{0}^{\epsilon_0} [r_{j}(x_1+\theta)+b_j] \mu'(\theta) \ud \theta - \int_{0}^{\epsilon_0} r_j y \mu'(\theta) \ud \theta\\
&&- \int_{0}^{\epsilon_0} [r_{j'}(x_1+\theta)+b_{j'}] \mu(\theta) \ud \theta + \int_{0}^{\epsilon_0} r_{j'} y \mu(\theta) \ud \theta\\
&=& \mu_2\frac{r_{j'}b_j-r_jb_{j'}}{b_j+r_j+b_{j'}+r_{j'}} > 0
\end{eqnarray*}
where $ \mu_2 = \int_0^{\epsilon_0} (1+y-x_1-\theta)\min\{p_j^*(x_1+\theta),p_{j'}^*(x_2+\theta)\}\ud \theta>0$.

When $x_2 < y \leq x_2+\epsilon_0:$
\begin{eqnarray*}
&&C(\mathbf p^*,y)-C(\mathbf p^1,y) \\
&=& \int_{0}^{\epsilon_0} [r_{j} (x_1+\theta)+b_{j}] \mu'(\theta) \ud \theta -\int_{0}^{y-x_2} [r_{j} (x_2+\theta)+b_{j}] \mu'(\theta) \ud \theta \\
&&- \int_{y-x_2}^{\epsilon_0} r_{j} y \mu'(\theta) \ud \theta -\int_{0}^{\epsilon_0} [r_{j'} (x_1+\theta)+b_{j'}] \mu(\theta) \ud \theta \\
&&+ \int_{0}^{y-x_2} [r_{j'} (x_2+\theta)+b_{j'}] \mu(\theta) \ud \theta
+\int_{y-x_2}^{\epsilon_0} r_{j'} y \mu(\theta) \ud \theta\\
&=& \mu_3 \frac{r_{j'}b_j-r_jb_{j'}}{b_j+r_j+b_{j'}+r_{j'}} > 0
\end{eqnarray*}
where $\mu_3 = \int_0^{y-x_2}(x_2-x_1)\min\{p_j^*(x_1+\theta),p_{j'}^*(x_2+\theta)\}\ud \theta+ \allowbreak
\int_{y-x_2}^{\epsilon_0}(1+y-x_1-\theta)\min\{p_j^*(x_1+\theta),p_{j'}^*
(x_2+\theta)\}\ud \theta$ .

When $y > x_2+\epsilon_0:$
\begin{eqnarray*}
&&C(\mathbf p^*,y)-C(\mathbf p^1,y) \\
&=& \int_{0}^{\epsilon_0} [r_{j} (x_1+\theta)+b_{j}] \mu'(\theta) \ud \theta -\int_{0}^{\epsilon_0} [r_{j} (x_2+\theta)+b_{j}] \mu'(\theta) \ud \theta \\
&& -\int_{0}^{\epsilon_0} [r_{j'} (x_1+\theta)+b_{j'}] \mu(\theta) \ud \theta + \int_{0}^{\epsilon_0} [r_{j'} (x_2+\theta)+b_{j'}] \mu(\theta) \ud \theta \\
&=& \mu_4 \frac{r_{j'}b_j-r_jb_{j'}}{b_j+r_j+b_{j'}+r_{j'}} > 0
 \end{eqnarray*}
where $ \mu_4 = \int_0^{\epsilon_0} (x_2-x_1)\min\{p_j^*(x_1+\theta),p_{j'}^*(x_2+\theta)\}\ud \theta>0$.

Thus we prove our earlier claim that $\mathbf p^1$ is not worse than $\mathbf p^*$ for any $y$. Since $\mathbf p^*$ is optimal, therefore $\mathbf p^1$ must be optimal as well.

However, we have shown that $\forall y \in (x_1+\epsilon_0, x_2)$, $C(\mathbf p^1,y_1)<C(\mathbf p^*,y_1)$, i.e.,  $\frac{C(\mathbf p^1,y)}{\mathrm{OPT}(y)}<\frac{C(\mathbf p^*,y)}{\mathrm{OPT}(y)}$, which contradicts Lemma~\ref{lemma:MSRstspace}. This completes the proof of the lemma.
\end{proof}

\subsection{Proof of Lemma~\ref{lemma:AlphaRelation}}
\begin{proof}
Plug $y\rightarrow d_{j}^-$ and $y\rightarrow d_{j}^+$ into (\ref{PfirstDiff}), we get
\begin{align*}
&\sum_{\tau=1}^{n} b_\tau p_\tau^*(d_{j}^-)+ \sum_{\tau=1}^{n} \int_{d_{j}^-}^{B}r_\tau p_\tau^*(x)\ud x=\lambda \\
&\sum_{\tau=1}^{n} b_{\tau}p_{\tau}^*(d_{j}^+)+ \sum_{\tau=1}^{n} \int_{d_{j}^+}^{B}r_\tau p_\tau^*(x)\ud x=\lambda
\end{align*}

Since Lemma~\ref{lemma:MSRfinite} shows that all the $p_j^*(x)$ are finite, we know $\sum_{\tau=1}^{n} \int_{d_{j}^-}^{B}r_\tau p_\tau^*(x)\ud x=\sum_{\tau=1}^{n} \int_{d_{j}^+}^{B}r_{\tau} p_{\tau}^*(x)\ud x$. Therefore, from the above 2 equations, we get $\sum_{\tau=1}^{n} b_\tau p_\tau^*(d_{j}^-)=\sum_{\tau=1}^{n} b_\tau p_\tau^*(d_{j}^+)$. From (\ref{PFormSolution}), we know $\sum_{\tau=1}^{n} b_\tau p_\tau^*(d_{j}^-)=b_j p_j^*(d_{j}^-)$ and $\sum_{\tau=1}^{n} b_\tau p_\tau^*(d_{j}^+)=b_{j-1}p_{j-1}^*(d_{j}^+)$, so we obtain
$$
\alpha_j b_j e^{r_jd_{j}/b_j} = \alpha_{j-1} b_{j-1} e^{r_{j-1}d_{j}/b_{j-1}}, \quad \forall j = 2,\cdots, n
$$
\end{proof}

\section{Computing the optimal strategy}

\subsection{Proof of Lemma~\ref{lemma:MSRconratio}}
\begin{proof}
For any $j=1,\cdots,n$, for any $y \in (d_{j+1},d_j]$, it holds that
\begin{eqnarray*}
C(\mathbf p,y) &=& \sum_{i=1}^n \left(\int_{0}^{y}(r_ix+b_i)p_j(x)\ud x+\int_{y}^{B}yr_i p_j(x)\ud x\right)\\
&=& \sum_{i=j+1}^n \int_{d_{i+1}}^{d_{i}}(r_i x + b_i)\alpha_i e^{\frac{r_i}{b_i}x}\ud x \\
&& + \int_{d_{j+1}}^{y}(r_jx+b_j)\alpha_j e^{\frac{r_j}{b_j}x}\ud x+\int_{y}^{d_j}yr_j \alpha_j e^{\frac{r_j}{b_j}x}\ud x\\
&&+ \sum_{i=1}^{j-1} \int_{d_{i+1}}^{d_{i}}y r_i\alpha_i e^{\frac{r_i}{b_i}x}\ud x\\
&=& \sum_{i=j+1}^n \alpha_i b_i\left(d_{i}e^{d_ie^{\frac{r_i}{b_i}d_i}}-d_{i+1}e^{\frac{r_i}{b_i}d_{i+1}}\right)\\
&&+\alpha_j b_j \left(ye^{\frac{r_j}{b_j}d_j}-d_{j+1}e^{\frac{r_j}{b_j}d_{j+1}}\right)
\\&&+ \sum_{i=1}^{j-1} y \alpha_i b_i \left(e^{\frac{r_i}{b_i}d_{i}}-e^{\frac{r_{i}}{b_{i}}d_{i+1}}\right)\\
&=& y \alpha_1 b_1 e^{\frac{b_1}{r_1}B}
\end{eqnarray*}
where the last equality holds because of (\ref{alphaRelation1}).
\end{proof}

\subsection{Proof of Lemma~\ref{lemma:dnconcave}}

\begin{proof}
We use the second derivative to prove it. Take the first derivative with repect to $d_n$:
\begin{eqnarray}
&& P'_{n-1}(d_n)\nonumber
\\&=& e^{\frac{r_{n-1}}{b_{n-1}}d_n}(\frac{r_{n-1}}{r_n}-1-(\frac{r_{n-1}}{r_n}-\frac{b_{n-1}}{b_n})e^{-\frac{r_n}{b_n}d_n})\label{pnfd}
\end{eqnarray}
Take the second derivative with respect to $d_n$:
\begin{eqnarray*}
&& P''_{n-1}(d_n)/\alpha_{n-1}
\\&=& e^{\frac{r_{n-1}}{b_{n-1}}d_n}\bigg[\frac{r_{n-1}}{b_{n-1}}(\frac{r_{n-1}}{r_n}-1)\\
&&-(\frac{r_{n-1}}{b_{n-1}}-\frac{r_n}{b_n})(\frac{r_{n-1}}{r_n}-\frac{b_{n-1}}{b_n})e^{-\frac{r_n}{b_n}d_n}\bigg]
\end{eqnarray*}
One can check that $\frac{r_{n-1}}{r_n}-1<0$, $(\frac{r_{n-1}}{b_{n-1}}-\frac{r_n}{b_n})<0$ and $(\frac{r_{n-1}}{r_n}-\frac{b_{n-1}}{b_n})<0$. Hence, $P''_{n-1}(d_n)<0$, and $P_{n-1}(d_n)$ is strictly concave.
\end{proof}

\subsection{Proof of Lemma~\ref{lemma:GConcave}}

\begin{proof}
Similarly, we take the first derivative with repect to $d_n$:
\begin{eqnarray}
&& P'_{j-1}(d_j)/\alpha_{j-1}\nonumber
\\&=& e^{\frac{r_{j-1}}{b_{j-1}}d_j}(\frac{r_{j-1}}{r_j}-1-(\frac{r_{j-1}}{r_j}-\frac{b_{j-1}}{b_j})e^{-\frac{r_j}{b_j}d_j})\nonumber\\
&&+\frac{D_j b_{j-1}}{b_j}(\frac{r_{j-1}}{b_{j-1}}-\frac{r_j}{b_j})e^{(\frac{r_{j-1}}{b_{j-1}}-\frac{r_j}{b_j})d_j}\label{pjfd}
\end{eqnarray}
Notice that $\frac{r_{j-1}}{b_{j-1}}-\frac{r_j}{b_j}<0$. So if $D_jr_j/b_j\ge 1$, we have:
\begin{eqnarray*}
&& P'_{j-1}(d_j)/\alpha_{j-1}
\\&\leq& e^{\frac{r_{j-1}}{b_{j-1}}d_j}(\frac{r_{j-1}}{r_j}-1-(\frac{r_{j-1}}{r_j}-\frac{b_{j-1}}{b_j})e^{-\frac{r_j}{b_j}d_j})\\
&&+\frac{ b_{j-1}}{r_j}(\frac{r_{j-1}}{b_{j-1}}-\frac{r_j}{b_j})e^{(\frac{r_{j-1}}{b_{j-1}}-\frac{r_j}{b_j})d_j}\\
&=& e^{\frac{r_{j-1}}{b_{j-1}}d_j}(\frac{r_{j-1}}{r_j}-1)<0\label{pjfd}
\end{eqnarray*}

If $D_jr_j/b_j< 1$, we take the second derivative with respect to $d_n$:
\begin{eqnarray*}
&& P''_{j-1}(d_j)/\alpha_{j-1}
\\&=& \frac{r_{j-1}}{b_{j-1}}(\frac{r_{j-1}}{r_j}-1)e^{\frac{r_{j-1}}{b_{j-1}}d_j}
\\&&-(\frac{r_{j-1}}{b_{j-1}}-\frac{r_j}{b_j})(\frac{r_{j-1}}{r_j}-\frac{b_{j-1}}{b_j})e^{(\frac{r_{j-1}}{b_{j-1}}-\frac{r_j}{b_j})d_j}\\
&&+\frac{D_j b_{j-1}}{b_j}(\frac{r_{j-1}}{b_{j-1}}-\frac{r_j}{b_j})^2 e^{(\frac{r_{j-1}}{b_{j-1}}-\frac{r_j}{b_j})d_j}\\
&=&\frac{r_{j-1}}{b_{j-1}}(\frac{r_{j-1}}{r_j}-1)e^{\frac{r_{j-1}}{b_{j-1}}d_j}\\
&&+(\frac{D_j b_{j-1}}{b_j}-\frac{b_{j-1}}{r_j})(\frac{r_{j-1}}{b_{j-1}}-\frac{r_j}{b_j})^2 e^{(\frac{r_{j-1}}{b_{j-1}}-\frac{r_j}{b_j})d_j}
\end{eqnarray*}
Easy to know that $\frac{r_{j-1}}{r_j}-1<0$ and $\frac{D_j b_{j-1}}{b_j}-\frac{b_{j-1}}{r_j}<0$. So $P''_{j-1}(d_j)<0$ and $P_{j-1}(d_j)$ is concave.
\end{proof}

\subsection{Proof of Lemma~\ref{lemma:invoc}}
\begin{proof}
By induction, when $j=n$, under the assumption that all the deletions in the past are correct, easy to know that the temporary $td_n$ makes $P'_{next[n]}(d_n)=0$. And $d_n^*<d_{next[j]}^*$ implies that $P'_{next[n]}(d_n^*)\le 0$. Also we know that $P'_{next[n]}(d_n=0)>0$ and $P_{next[n]}(d_n)$ is a concave function. So the optimal $d_n^*$ must make $P'_{next[n]}(d_n)=0$, i.e. $d_n^*=td_n$.

When $j<n$, there may be deletions during an invocation of $ComputingBP(j)$. So we split into 2 cases:

\textbf{Case 1}: Shop $j$ is not deleted in this invocation.

It means that the temporary $td_j>td_{prev[j]}$ and no deletion happens during this invocation. Under the assumption that all the deletions in the past are correct, it is easy to verify that the temporary breakpoints $td_n,td_{n-1},\cdots,td_j$ have maximized $P_{next[j]}$ when $d_{next[j]}^*>td_j$.

But if $d_{next[j]}^*\le td_j$, we can prove the solution of the alive breakpoints $\mathbf{td^*}= (td_n^*,td_{next[n]}^*,\cdots,td_j^*,d_{next[j]}^*)$ satisfying that $\forall \tau, td_\tau^*=\min\{td_\tau, d_{next[j]}^*\}$ is better than any other solution $\mathbf{d}=(d_n,d_{next[n]},\cdots,d_j,d_{next[j]}^*)$ satisfying $\forall \tau,d_\tau\le d_{next[j]}^*$.
Notice that all the deletions in the past are supposed to be correct.

To prove it, we extend the feasible solution field from $\{\mathbf{d}=(d_n,d_{next[n]},\cdots,d_{next[j]}):d_{\tau_1}\ge d_{\tau_2}\text{ if }\tau_1<\tau_2\}$ to $\{\mathbf{d}=(d_n,d_{next[n]},\cdots,d_{next[j]}):d_\tau\le d_{next[j]}\}$, i.e. we relax the constraint for the order of the breakpoints and want to prove the optimality of $\mathbf{td^*}$ in the relaxed field which is a stronger proposition. We redefine $$P_{next[j]}\triangleq\sum_{\tau\ge next[j],\tau\text{ is alive}} (\int_{d_{\tau+1}}^{d_\tau} p_\tau(x)dx)$$ and for any 2 solutions of the breakpoints $\mathbf{d_1}$ and $\mathbf{d_2}$, we define $$S_{next[j]}\triangleq P_{next[j]}/\alpha_{next[j]}$$
Consider that all the breakpoints are at $\mathbf{td}$ initially, we move them one by one in the order from $td_j$ to $td_n$.

We construct a movement sequence $$\mathbf{td^{(\eta)}}= (td_n^{(\eta)},td_{next[n]}^{(\eta)},\cdots,td_j^{(\eta)},d_{next[j]}^*)$$ satisfying $$\forall\tau>\eta, td_\tau^{(\eta)}=td_\tau;\forall\tau\le\eta, td_\tau^{(\eta)}=\min\{td_\tau,d_{next[j]}^*\}$$ and obviously $\mathbf{td}^{(n)}=\mathbf{td^*},\mathbf{td}^{(next[j])}=\mathbf{td}$. Also we define $$\eta^*=\min_\eta \eta,\text{ s.t. } td_\eta<d_{next[j]}^*$$

Similarly we construct another sequence for the movement from $\mathbf{td}$ to $\mathbf{d}$ $$\mathbf{d^{(\eta)}}=(d_n^{(\eta)},d_{next[n]}^{(\eta)},\cdots,d_j^{(\eta)},d_{next[j]}^*)$$ satisfying $$\forall\tau>\eta, d_\tau^{(\eta)}=td_\tau;\forall\tau\le\eta, td_\tau^{(\eta)}=d_\tau$$ and easy to see $\mathbf{d}^{(n)}=\mathbf{d},\mathbf{d}^{(next[j])}=\mathbf{td}$.

By induction, first we have $\mathbf{d}^{(next[j])}=\mathbf{td}^{(next[j])}$. Then the proof is divided into 2 steps.

\textbf{Step 1}: when $\eta<\eta^*$, we assume that $$\alpha_{next[\eta]}(\mathbf{td}^{(next[\eta])})\le \alpha_{next[\eta]}(\mathbf{d}^{(next[\eta])})$$ and $$P_{next[j]}(\mathbf{td}^{(next[\eta])})\ge P_{next[j]}(\mathbf{d}^{(next[\eta])})$$

And we want to prove that $\alpha_{\eta}(\mathbf{td}^{(\eta)})\le \alpha_{\eta}(\mathbf{d}^{(\eta)})$ and $P_{next[j]}(\mathbf{td}^{(\eta)})\ge P_{next[j]}(\mathbf{d}^{(\eta)})$. According to the definition of the sequence and Lemma~\ref{lemma:AlphaRelation}, easy to know that $$\alpha_{next[\eta]}(\mathbf{td}^{(\eta)})=\alpha_{next[\eta]}(\mathbf{td}^{(next[\eta])})$$ $$\alpha_{next[\eta]}(\mathbf{d}^{(\eta)})=\alpha_{next[\eta]}(\mathbf{d}^{(next[\eta])})$$
So $\alpha_{next[\eta]}(\mathbf{td}^{(\eta)})\le \alpha_{next[\eta]}(\mathbf{d}^{(\eta)})$. Since $\eta<\eta^*$, $td_\eta^{(\eta)}=d_{next[j]}^*\ge d_\eta^{(\eta)}$. Then derived from Lemma~\ref{lemma:AlphaRelation}, we get
$$\frac{\alpha_{\eta}(\mathbf{td}^{(\eta)})}{\alpha_{next[\eta]}(\mathbf{td}^{(\eta)})} =\frac{b_{next[\eta]}}{b_\eta}\exp((\frac{r_{next[\eta]}}{b_{next[\eta]}}-\frac{r_\eta}{b_\eta})td_\eta^{(\eta)})$$
$$\frac{\alpha_{\eta}(\mathbf{d}^{(\eta)})}{\alpha_{next[\eta]}(\mathbf{d}^{(\eta)})} =\frac{b_{next[\eta]}}{b_\eta}\exp((\frac{r_{next[\eta]}}{b_{next[\eta]}}-\frac{r_\eta}{b_\eta})d_\eta^{(\eta)})$$
Notice that $\frac{r_{next[\eta]}}{b_{next[\eta]}}-\frac{r_\eta}{b_\eta}<0$, so we have $\frac{\alpha_{\eta}(\mathbf{td}^{(\eta)})}{\alpha_{next[\eta]}(\mathbf{td}^{(\eta)})}\le\frac{\alpha_{\eta}(\mathbf{td}^{(\eta)})}{\alpha_{next[\eta]}(\mathbf{td}^{(\eta)})}$ and also $\alpha_{\eta}(\mathbf{td}^{(\eta)})\le \alpha_{\eta}(\mathbf{d}^{(\eta)})$.

Now we begin to prove $P_{next[j]}(\mathbf{td}^{(\eta)})\ge P_{next[j]}(\mathbf{d}^{(\eta)})$ and we just need to prove $P_{next[j]}(\mathbf{td}^{(next[\eta])})-P_{next[j]}(\mathbf{td}^{(\eta)})\le P_{next[j]}(\mathbf{d}^{(next[\eta])})-P_{next[j]}(\mathbf{d}^{(\eta)})$. $\forall\tau>\eta$, we fix each breakpoint $td_\tau^{(\eta)}$ at position $td_\tau$ and $\forall\tau<\eta$, we fix each breakpoint $td_\tau^{(\eta)}$ at position $d_{next[j]}^*$. Then the multivariate function $S_{next[\eta]}(\mathbf{td}^{(\eta)})$ can be considered as a univariate function $S_{next[\eta]}(td_\eta^{(\eta)})$. Since we've known $\alpha_{next[\eta]}(\mathbf{td}^{(\eta)})=\alpha_{next[\eta]}(\mathbf{td}^{(next[\eta])})$, we get the following equation
\begin{eqnarray*}
&&P_{next[j]}(\mathbf{td}^{(next[\eta])})-P_{next[j]}(\mathbf{td}^{(\eta)})\\&=&P_{next[\eta]}(\mathbf{td}^{(next[\eta])})-P_{next[\eta]}(\mathbf{td}^{(\eta)})\\&=& \alpha_{next[\eta]}(\mathbf{td}^{(\eta)})(S_{next[\eta]}(td_\eta)-S_{next[\eta]}(td_\eta^{(\eta)}))
\end{eqnarray*}
It is easy to verify that the breakpoints whose indexes are less than $\eta$ are irrelevant to the difference value of function $S_{next[\eta]}$ when it is known that those breakpoints in the 2 tuples are identical. Therefore
\begin{eqnarray*}
&& P_{next[j]}(\mathbf{d}^{(next[\eta])})-P_{next[j]}(\mathbf{d}^{(\eta)})\\&=& P_{next[\eta]}(\mathbf{d}^{(next[\eta])})-P_{next[\eta]}(\mathbf{d}^{(\eta)})\\&=& \alpha_{next[\eta]}(\mathbf{d}^{(\eta)})(S_{next[\eta]}(td_\eta)-S_{next[\eta]}(d_\eta^{(\eta)}))
\end{eqnarray*}
Since we have known $\alpha_{next[\eta]}(\mathbf{td}^{(\eta)})=\alpha_{next[\eta]}(\mathbf{td}^{(next[\eta])})\le \alpha_{next[\eta]}(\mathbf{d}^{(next[\eta])})= \alpha_{next[\eta]}(\mathbf{d}^{(\eta)})$ and $td_\eta$ is the maximal point of the concave function $S_{next[\eta]}(\cdot)$, we only need to prove that $S_{next[\eta]}(td_\eta^{(\eta)})\ge S_{next[\eta]}(d_\eta^{(\eta)})$. According to Lemma~\ref{GConcave} and the algorithm, we know $S_{next[\eta]}(d_\eta)$ is concave and $td_\eta$ is the maximal point. Since $\eta<\eta^*$ and $td_\eta\ge td_\eta^{(\eta)}=d_{next[j]}^*\ge d_\eta^{(\eta)}$, we get $S_{next[\eta]}(td_\eta^{(\eta)})\ge S_{next[\eta]}(d_\eta^{(\eta)})$. Step 1 is done.

\textbf{Step 2}: when $\eta\ge\eta^*$, we assume that
$$P_{next[j]}(\mathbf{td}^{(next[\eta])})\ge P_{next[j]}(\mathbf{d}^{(next[\eta])})$$

And we want to prove that $P_{next[j]}(\mathbf{td}^{(\eta)})\ge P_{next[j]}(\mathbf{d}^{(\eta)})$. Similarly, we just need to prove that $P_{next[j]}(\mathbf{td}^{(next[\eta])})-P_{next[j]}(\mathbf{td}^{(\eta)})\le P_{next[j]}(\mathbf{d}^{(next[\eta])})-P_{next[j]}(\mathbf{d}^{(\eta)})$. But $\mathbf{td}^{(next[\eta])}=\mathbf{td}^{(\eta)}$ because $\eta\ge\eta^*$ and similarly we have the following equation
\begin{eqnarray*}
&& P_{next[j]}(\mathbf{d}^{(next[\eta])})-P_{next[j]}(\mathbf{d}^{(\eta)})\\&=& \alpha_{next[\eta]}(\mathbf{d}^{(\eta)})(S_{next[\eta]}(td_\eta)-S_{next[\eta]}(d_\eta^{(\eta)}))
\end{eqnarray*}
So we only need to prove $S_{next[\eta]}(td_\eta)\ge S_{next[\eta]}(d_\eta^{(\eta)})$ which is obvious since $td_\eta$ is the maximal point of this concave function. The induction inside is done and we have proved that $P_{next[j]}(\mathbf{td}^*)\ge P_{next[j]}(\mathbf{d})$. Notice that once $\exists\tau$ such that $d_\tau\ne td_\tau^*$, the equation does not hold. It means $\mathbf{d}^*=\mathbf{td}^*$. We get $d_j^*=d_{next[j]}^*$ if $d_{next[j]}^*\le td_j$, so we always have $d_{next[j]}^*> td_j$ if $d_{next[j]}^*> d_j^*$. Case 1 is proved.

\textbf{Case 2}: Shop $j$ is deleted in this invocation.
This case happens when $D_jr_j/b_j\ge 1$ or $td_{j}\le td_{prev[j]}$. First we need to prove that $d_{prev[j]}^*=d_j^*$, i.e. the deletion is correct. By reductio ad absurdum, we assume that $d_{prev[j]}^*<d_j^*$. According to the inductive assumption, $\forall \tau>j$ we must have $td_\tau=d_\tau^*$. Fixing all the other breakpoints except $d_j$ at $\mathbf{d}^*$, $d_{j}^*$ is to maximize the function $P_{next[j]}(d_j)$. From Lemma~\ref{GConcave}, we know that $\forall d_j,P'_{next[j]}(d_j)<0$ if $D_jr_j/b_j\ge 1$. And $td_j$ is the maximal point of the concave function $P_{next[j]}(d_j)$ if $td_{j}\le td_{prev[j]}$. It means $\forall d_j^*>d_{prev[j]}^*=td_{prev[j]}$, $P'_{next[j]}(d_j^*)<0$ which makes contradicticon to that $d_j^*$ is optimal. So we proved $d_{prev[j]}^*=d_j^*$.

After the deletion, our algorithm invokes the function $ComputingBP(prev[j])$ and the same property still holds because of the inductive assumption. Case 2 is proved.
\end{proof}

\section{MSR-S}

Now we prove lemma \ref{lemma:MSR-S}.
\begin{proof}
Given an action $\psi$, we assume that the buying time in $\psi$ is $x$ and there exists a switching operation strictly before $x$. Now we need to prove that $\psi$ is dominated.

We assume that the first switching operation in $\psi$ is from shop $i$ to shop $j$ at time $x_0$ ($0\le x_0<x$). If there are more than one switching operations at time $x_0$, they can be considered as a big switching operation. But the switching cost is the sum of all the switching cost at time $x_0$. We assume that the last switching operation at $x_0$ is from some shop to shop $k$. If the switching cost at $x_0$ is larger than $c_{ik}^*$, then obviously it is dominated. But if not, then there will be 2 cases.

\textbf{Case 1 :} when $r_i< r_k$,

Denote by $x_1$ the switching time just after $x_0$. But if switching only happens at time $x_0$ in $\psi$, we let $x_1$ to be $x$. It can be seen that $x_0<x_1$. Since $r_i\le r_k$, in action $\psi$ it is better for the consumer to move the switching operation(s) from time $x_0$ to time $x_1$. It can be verified that the cost remains the same when $y<x_0$ and decreases when $y\ge x_0$. Thus, $\psi$ is dominated.

\textbf{Case 2 :} when $r_i\ge r_k$,

Consider another action $\psi^1$: the consumer enters shop $k$ at the very beginning and does not execute any switching operation at time $x_0$, but she acts the same as $\psi$ after time $x_0$. It can be verified that the cost does not increase when $y<x_0$ and decreases when $y\ge x_0$. Thus, it is dominated by $\psi^1$.
\end{proof}

\section{MSR-E}

\subsection{Proof of Lemma \ref{lemma:EFstspace}}
\begin{proof}
Similar to the proof of lemma~\ref{lemma:MSRstspace}, we first prove that $y\in[B,+\infty)$ is dominated by $y=+\infty$. Then we prove $\forall j\in[n]$ and $x\in(B,+\infty)\cup\{+\infty\}$, $(j,x)$ is dominated by $(j,B)$.
\end{proof}

\subsection{Proof of Lemma \ref{lemma:MSR-E constant}}
\begin{proof}
The proof is very similar to the proof of lemma~\ref{lemma:MSRconstant}. But in case (a), we need to modify the construction of $\mathbf{p}^1$ a little bit. Recall that we moved some probability from $p_{j,y}$ to $p_{1,B}$, but now we move it from $p_{j,y}$ to $p_{k,B}$ where $k$ is defined as $\arg\min_k a_k+r_k y$. Then, similarly, we prove this lemma.
\end{proof}

\subsection{Proof of Lemma \ref{lemma:MSR-E finite}}
\begin{proof}
Similar to the proof of lemma~\ref{lemma:MSRfinite}.
\end{proof}

\subsection{Proof of Lemma~\ref{lemma:moveP MSR-E}}
\begin{proof}
The proof is very similar to the proof of lemma~\ref{lemma:MSR-E constant}. But in case (b), we need to modify the construction of $\mathbf{p}^1$ a little bit. Recall that we moved some probability from $p_{j,y}$ to $p_{1,B}$, but now we move it from $p_{j,y}$ to $p_{k,B}$ where $k$ is defined as $\arg\min_k a_k+r_k y$. Then, similarly, we prove this lemma.
\end{proof}

\section{MSR-ES}
\subsection{Proof of Lemma~\ref{lemma:EAdominate}}
\begin{proof}
Similar with the proof of Lemma~\ref{lemma:MSRstspace}, we first prove the dominance for the nature's strategies. Suppose $c(\psi,y)$ to be the cost when the consumer choose a pure strategy $\psi$ and the nature choose a pure strategy $y$. Easy to know that $\forall y,c(\psi,y)\le c(\psi,+\infty)$, then $\forall y\in [B,+\infty)$, the ratio $R(\psi,y)$ satisfies that
$$R(\psi,y)=\frac{c(\psi,y)}{\min_j\{a_j+b_j\}}\le\frac{c(\psi,+\infty)}{\min_j\{a_j+b_j\}}=R(\psi,+\infty)$$

Then we prove the dominance for the consumer's strategies. For a pure strategy $\psi$, if the consumer does not buy at or before time $B$, then we construct a strategy $\psi^1$, buying at time $B$, to dominate $\psi$.
The construction is as follows:

Delete all the event $(j,x)$, satisfying $j\ge 0, x\ge B$, from the strategy $\psi$. Then add the event $(0,B)$ and we obtain the strategy $\psi^1$.

Since the nature only chooses $y\in(0,B)\cup\{+\infty\}$, it can be verified that $\psi^1$ performs the same with $\psi$ if $y\in(0,B)$. And notice that the assumption at the beginning of Section 4: $\forall i,j, b_i\le a_j+b_j$, therefore, $\psi^1$ performs better if $y=+\infty$. The dominance for the consumer's strategies is also proved. And easy to verify that there is no difference between $y=+\infty$ and $y=B$ if the consumer buys in $[0,B]$.
\end{proof}
\subsection{Proof of Lemma~\ref{lemma:reduceSpace}}
\begin{proof}
Similarly with the proof of Lemma~\ref{lemma:MSR-S}, we prove the dominance for the following 3 conditions:
\begin{itemize}
\item $\exists 0<\tau<|\psi|-1$ such that $x_{\tau-1}=x_\tau$:
\end{itemize}
It can be verified that $\psi$ is dominated by $\psi^{(2)}$, constructed in this way:

Delete event $(j_{\tau-1},j_{\tau},x_\tau)$ from $\psi$ and reset the value of $j_{\tau-1}$ to be $j_\tau$. Thus we get $\psi^{(2)}$.
\begin{itemize}
\item $\exists (i,j,x)\in \psi$ such that $r_i\le r_j$ and $(j,0,x)\notin\psi$:
\end{itemize}
From Lemma~\ref{lemma:EAdominate}, we know $(i,j,x)$ is not the last event in $\psi$ since the the last one is buying event. We suppose the event just after $(i,j,x)$ to be $(j,k,x')$ ($k$ can be 0 if $x\neq x'$). It can be verified that $\psi$ is dominated by $\psi^{(2)}$, constructed in this way:

If $x=x'$, it is the same with the first condition.

If $x<x'$ and $k\neq 0$, delete event $(i,j,x)$ from $\psi$ and reset the event $(j,k,x')$ to be $(i,k,x')$.

If $x<x'$ and $k=0$, reset the event $(i,j,x)$ to be $(i,j,x')$ in $\psi$. Thus we get $\psi^{(2)}$.
\begin{itemize}
\item $\exists (i,j,x)\in \psi$ such that $a_i\ge a_j$ and $(j,0,x)\notin\psi$:
\end{itemize}
We suppose the event just before $(i,j,x)$ to be $(k,i,x')$ ($k$ can be 0). It can be verified that $\psi$ is dominated by $\psi^{(2)}$, constructed in this way:

If $x=x'$, it is the same with the first condition.

If $r_i\le r_j$, it is the same with the second condition.

If $x>x'$ and $r_i>r_j$, delete event $(i,j,x)$ from $\psi$ and reset the event $(k,i,x')$ to be $(k,j,x')$. Thus we get $\psi^{(2)}$.
\end{proof}
\subsection{Proof of Lemma~\ref{lemma:EAconstant}}
\begin{proof}
The proof is very similar to the proof of lemma~\ref{lemma:MSRconstant}. But in case (b), we need to modify the construction a little bit. Recall that we proved that buying operation appears in that particular interval with a positive probability, but now it can be either buying operation or switching operation and we bring the operation forward in the same way (There may also be some difference with those parameters which doesn't matter). Then, similarly, we prove this lemma.
\end{proof}

\subsection{Proof of Lemma~\ref{lemma:EAf2p}}

\begin{proof}
Since the equation~(\ref{eqn:f2q}b) in the following theorem is strictly stronger than this lemma, we leave the it in the proof of theorem~\ref{theorem:MSR-ES}.
\end{proof}

\subsection{Proof of Theorem~\ref{theorem:MSR-ES}}

\begin{proof}
We first prove the equation~(\ref{eqn:f2q}b). To prove it, we define the virtual cost in a virtual shop $(i,j)$ (corresponding to a switching or buying operation) as follows:
\begin{align*}
c_{(i,j)}(x,y)\triangleq
\begin{cases}
a_{(i,j)}+r_{(i,j)}y,&\text{ if }y<x;\\
a_{(i,j)}+r_{(i,j)}x+b_{(i,j)},\quad&\text{ if }y\ge x.
\end{cases}
\end{align*}

And for any $\mathbf{s}\in\mathcal{S}$ in the short form, we denote the $(i+1)$th item of $\mathbf{s}$ by $j_i^{(\mathbf{s})}$, i.e.,
$$\mathbf{s}=\{j_0^{(\mathbf{s})},j_1^{(\mathbf{s})},\cdots,0\}$$

Similarly, for any $\mathbf{x}\in\mathcal{X}$, we denote the $i$th item of $\mathbf{x}$ by $x_i^{(\mathbf{x})}$, i.e.,
$$\mathbf{x}=\{x_1^{(\mathbf{x})},x_2^{(\mathbf{x})},\cdots\}$$

Then we know by definition that:
\begin{align*}
c(\psi,y)&\triangleq
\begin{cases}
\sum_{\tau=0}^{k-1} [a_{j_\tau^{(\mathbf{s}(\psi))}}+r_{j_\tau^{(\mathbf{s}(\psi))}} (x_{\tau+1}^{(\mathbf{x}(\psi))}-x_{\tau}^{(\mathbf{x}(\psi))})]\\+r_{j_{k}^{(\mathbf{s}(\psi))}}(y-x_k^{(\mathbf{x}(\psi))}),\\
\quad\quad\quad\quad\text{ if }\exists 0<k<|\psi|, x_{k-1}\le y<x_k;\\
\sum_{\tau=0}^{|\psi|-2} (a_{j_\tau^{(\mathbf{s}(\psi))}}+r_{j_\tau^{(\mathbf{s}(\psi))}} (x_{\tau+1}^{(\mathbf{x}(\psi))}-x_{\tau}^{(\mathbf{x}(\psi))}))+b_{j_{|\psi|-2}^{(\mathbf{s}(\psi))}},\\
\quad\quad\quad\quad\text{ if }y\ge x_{|\psi|-1}.
\end{cases}\\
&=\sum_{\tau=1}^{|\psi|-1}c_{(j_{\tau-1}^{(\mathbf{s}(\psi))},j_{\tau}^{(\mathbf{s}(\psi))})}(x_{\tau}^{(\mathbf{x}(\psi))},y)
\end{align*}
Then we get
\begin{eqnarray*}
C(\mathbf{f}, y) &\triangleq& \sum_{\mathbf{s}\in\mathcal{S}}\idotsint\limits_{\mathbf{x}\in\mathcal{X}_\mathbf{s}}c_{\mathbf{s}}(\mathbf{x},y) f_{\mathbf{s}}(\mathbf{x})\ud\mathbf{x}\\
&=&\sum_{\mathbf{s}\in\mathcal{S}}\idotsint\limits_{\mathbf{x}\in\mathcal{X}_\mathbf{s}}\sum_{\tau=1}^{|\psi|-1}c_{(j_{\tau-1}^{(\mathbf{s})},j_{\tau}^{(\mathbf{s})})}(x_{\tau}^{(\mathbf{x})},y) f_{\mathbf{s}}(\mathbf{x})\ud\mathbf{x}\\
&=&\sum\limits_{(i,j)\text{ is a virtual shop}} \int_{0}^{B}c_{(i,j)}(x,y)p_{(i,j)}^{(\mathbf{f})}(x)\ud x\\
&=&\sum\limits_{(i,j)\text{ is a virtual shop}} C_{(i,j)}(\mathbf{p^{(f)}},y)\\
&=&{\sum_{(i,j)\in[n]^2:a_i<a_j,r_i>r_j}C_{(i,j)}(\mathbf{p^{(f)}},y)}\\&&+{\sum_{j\in[n]}C_{(j,0)}(\mathbf{p^{(f)}},y)}
\end{eqnarray*}

Thus, the equation~(\ref{eqn:f2q}b) has been proved. And equation~(\ref{eqn:f2q}a) has been proved in lemma~\ref{lemma:EAconstant}. Equations~(\ref{eqn:f2q}c) and (\ref{eqn:f2q}d) are sufficient and necessary constraints to ensure that, when we find a $\mathbf{p}$ for the virtual shops, there exists a legal strategy $\mathbf{f}$ for the consumer such that $\mathbf{p^(f)}=\mathbf{p}$.
\end{proof}

\subsection{Proof of Lemma~\ref{lemma:MSR-ES seg}}
\begin{proof}
Similar to the proofs of lemma~\ref{lemma:MSR-E finite} and \ref{lemma:moveP MSR-E}.
\end{proof}

\end{document}